\documentclass[10pt]{article}
\usepackage[colorlinks=true, pdfstartview=FitV, linkcolor=black, citecolor=black,urlcolor=black]{hyperref}

\usepackage{amsmath,amssymb,cancel,color,framed}
\definecolor{shadecolor}{rgb}{0.9, 0.9, 0.81}
\usepackage{amsthm}

\let\oldmarginpar\marginpar
\renewcommand\marginpar[1]{\-\oldmarginpar[\raggedleft\footnotesize #1]%
{\raggedright\footnotesize #1}}

\newtheorem{thm}{Theorem}[section]
\newtheorem{prop}{Proposition}[section]
\newtheorem{defn}{Definition}[section]
\newtheorem{lemma}{Lemma}[section]
\newtheorem{assumption}{Assumption}[section]

\newtheorem{remark}{Remark}[section]

\newtheorem{example}{Example}[section]
\newcommand{\bpm}{\begin{pmatrix}}
\def\epm{\end{pmatrix}}

\def\le{\left}
\def\ri{\right}
\def\be{\begin{equation}}
\def\ee{\end{equation}}
\def\bea{\begin{eqnarray}}
\def\eea{\end{eqnarray}}
\def\bp{\begin{prop}}
\def\ep{\end{prop}}
\def \pa{\partial}

\def\1{{\bf 1}}

\def\bs{\boldsymbol}
\def\d{{\rm d} }
\def\eq{\begin{equation}}
\def\endeq{\end{equation}}
\def\eqarr{\begin{eqnarray}}
\def\endeqarr{\end{eqnarray}}
\def\eqnn{\begin{equation*}}
\def\endeqnn{\end{equation*}}
\def\ds{\displaystyle}

\newcommand{\C}{\mathbb{C}}

\newcommand{\R}{\mathbb{R}}

\usepackage{color, marvosym}
\definecolor{light-blue}{rgb}{0.8,0.85,1}
\definecolor{blue}{rgb}{0,0,1}
\definecolor{red}{rgb}{1,0,0}

\renewcommand{\theequation}{\arabic{section}-\arabic{equation}}
\makeatletter
\@addtoreset{equation}{section}
\makeatother
\begin{document}

\baselineskip 16pt plus 1pt minus 1pt

\title{{\bf Spectra of random Hermitian matrices with a small-rank external source:  supercritical and subcritical regimes}}

\author{Bertola, Buckingham, Lee, Pierce}

\author{{\bf M. Bertola}\footnote{Department of Mathematics and Statistics, Concordia University, Montr\'eal, Qu\'ebec H3G1M8 and Centre de recherches math\'ematiques, Universit\'e de Montr\'eal, Qu\'ebec H3T1J4;  bertola@mathstat.concordia.ca},
{\bf R. Buckingham}\footnote{Department of Mathematical Sciences, University of
Cincinnati, Ohio 45221; buckinrt@uc.edu},
{\bf S. Y. Lee}\footnote{Department of Mathematics, California Institute of Technology, Pasadena, California 91125; duxlee@caltech.edu},
{\bf V. Pierce}\footnote{Department of Mathematics, University of Texas -- Pan American, Edinburg, Texas, 78539; piercevu@utpa.edu}}
\date{\today}
\maketitle

\begin{abstract}
Random Hermitian matrices with a source term arise, for instance, in the study of non-intersecting Brownian walkers \cite{Adler:2009a, Daems:2007} and sample covariance matrices \cite{Baik:2005}.
 We consider the case when the $n\times n$ external 
source matrix has two distinct real eigenvalues:  $a$ with multiplicity 
$r$ and zero with multiplicity $n-r$.  The source is small in the 
sense that $r$ is finite or $r=\mathcal O(n^\gamma)$, for $0<  \gamma<1$.  For a 
Gaussian potential, P\'ech\'e \cite{Peche:2006} showed that for $|a|$ 
sufficiently small (the subcritical regime) the external source has 
no leading-order effect on the eigenvalues, while for $|a|$ sufficiently 
large (the supercritical regime) $r$ eigenvalues exit the bulk of the 
spectrum and behave as the eigenvalues of $r\times r$ Gaussian unitary ensemble (GUE).  
We establish the universality of these results for a general class 
of analytic potentials in the supercritical and subcritical regimes.
\end{abstract}

\tableofcontents
\section{Introduction}

The physical motivation
behind studying Hermitian random matrix ensembles is as a
model of the Hamiltonian for complex systems, where the eigenvalues of the
random Hermitian matrix represent the energy levels of a system without
time-reversal invariance \cite{Wigner:1951}.

Let ${\bf A}$ be a fixed Hermitian matrix.  We consider the set of all
$n\times n$ Hermitian matrices ${\bf M}$ endowed with the probability measure
\eq
\label{prob-measure}
\mu_n(d{\bf M}) = \frac{1}{Z_n}e^{-n\text{Tr}(V({\bf M})-{\bf A}{\bf M})}d{\bf M}; \quad Z_n:=\int e^{-n\text{Tr}(V({\bf M})-{\bf A}{\bf M})}d{\bf M},
\endeq
where $d{\bf M}$ is the entry-wise Lebesgue measure
and the integration is over all Hermitian matrices.  

When ${\bf A}={\bs 0}$ (no external source) and 
$V({\bf M})={\bf M}^2/2$, \eqref{prob-measure}
describes the Gaussian Unitary Ensemble, or GUE.  For ${\bf A}\neq {\bf 0}$ 
and $V({\bf M})={\bf M}^2/2$, this measure arises in the study of 
Hamiltonians that can be written as the sum of a random matrix and a 
deterministic source matrix \cite{Brezin:1996}.
We are specifically interested in small-rank sources of the form
\eq
\label{A}
{\bf A} = \text{diag}(\underbrace{a,\dots,a}_r,\underbrace{0,\dots,0}_{n-r})
\endeq
assuming that either $r=\mathcal O(n^\gamma)$, $0<\gamma<1$ or $r$ is finite 
(in which case we define $\gamma:=0$).  The ratio of $r$ to $n$, which is 
asymptotically small, will be denoted as 
\eq
\kappa:=\frac{r}{n}.
\endeq
P\'ech\'e \cite{Peche:2006} studied the limiting distribution of the largest
eigenvalue in the Gaussian case ($V({\bf M})={\bf M}^2/2$) under these 
assumptions and found three distinct behaviors.  In
the supercritical case, $r$ eigenvalues are expected to exit the
bulk and are found to distribute as the eigenvalues of an
$r\times r$ GUE matrix.  For the subcritical case, the
largest eigenvalue is expected to lie at the right band endpoint and
behave as the largest eigenvalue of an $n\times n$ GUE matrix.  
In the critical case, when the outliers
lie at the band endpoint, the distribution for the largest eigenvalue
is an extension of the standard GUE Tracy-Widom function
\cite{Tracy:1994} which arises when $r=0$ (see also \cite{Adler:2009a, Baik:2005,Baik:2006}).

One of the primary goals of random matrix theory is to determine universality
classes of matrix ensembles, that is, find different probability measures
on the space of matrices for which the spectral properties are the same in
the large-$n$ limit.  With this goal in mind we consider more general
functions $V({\bf M})$, with our specific assumptions listed in Section 
\ref{assumptions}.  Basically, we assume $V({\bf M})$ is a generic 
single-gap analytic potential with sufficient growth at infinity.  We 
show that P\'ech\'e's results \cite{Peche:2006} hold for these general potentials in the 
supercritical and subcritical cases.  The universality of the critical case 
will be considered elsewhere \cite{Bertola:2010}.


\subsection{Definition of the supercritical, subcritical, and critical regimes}
Let $g(z)$ be the $g$-function associated with the orthogonal
polynomials with potential $V(z)$ (see, for instance, \cite{Deift:1998-book} 
or \cite{Deift:1999a}).  
It may be written as
\eq
\label{g-form}
g(z):=\int_{\mathbb{R}}\log(z-s)\rho_\text{min}(s)ds,
\endeq
where $\rho_\text{min}$ is the unique probability measure minimizing the 
functional
\eq
\mathcal{F}[\rho]:=\int_\mathbb{R}V(s)\rho(s)ds-\int_\mathbb{R}\int_\mathbb{R}\rho(s)\rho(s')\log|s-s'|dsds'.
\endeq 
We will assume this equilibrium measure $\rho_\text{min}$ is supported on 
a single band $[\alpha,\beta]$ (see Assumption \ref{assumptionAV} (v) in Section 
\ref{assumptions}).  Define
\begin{eqnarray}
\label{P1-nolog}
P_1(z) & := & -V(z)+2g(z)+l_1, \\
\label{P2-nolog}
P_2(z) & := & -V(z)+az+g(z)+l_2, \\
\label{P3-nolog}
P_3(z) & := & -P_1(z) + P_2(z) = az-g(z)-l_1+l_2.
\end{eqnarray}
Here $g(z)$ and $l_1$ are uniquely determined by the conditions
that $P_1(z)_\pm$ is purely imaginary on 
{the support of the equilibrium measure}  
and has negative real part on its complement in $\R$ (here the subscripts $_\pm$ denote the boundary values from above/below the real axis).  How 
$l_2$ is chosen will be described at the end of this section.

It is also known that $\Re g(z)$ is a  {\bf continuous} function on $\R$ and {\bf harmonic} on the complement of the support of $\rho_{\text{min}}$ (up to a sign it is also known as  the {\em logarithmic potential} in potential theory).

\begin{defn}
\label{acrit}
Define $a_c$ to be the (unique) value of $a$ so that $P_2'(\beta)=0$. 
\end{defn}
The uniqueness is promptly seen because $P_2'(\beta) = -V'(\beta) + a +g'(\beta)$; 
in fact the effective potential $P_1$  is known \cite{Deift:1998a} to satisfy 
\be
P_1'(z)= \mathcal O{(z-\beta)^\frac 12}
\ee
and in particular $P_1'(\beta)=0$ and hence $P_2'(\beta) = a-g'(\beta)  = a-\frac 12 V'(\beta)$. Thus the critical value of $a$ is given by
\begin{equation} 
a_c = g'(\beta) = \frac{1}{2} V'(\beta)\,.\label{110}
 \end{equation}
We first have 
\begin{lemma}
\label{acpositive}
The critical value $a_c = g'(\beta)$ is positive. Moreover $g'(\alpha)<0$.
\end{lemma}
\begin{proof}
From \eqref{g-form} we see that $g'(z) = \int_{\alpha}^\beta \frac 1{z-s} \rho_{\text{min}}(s)\d s$ is {\em positive} for $z>\beta$. It is also known that the density $\rho_{\text{min}}$ vanishes like a square root at the endpoints $\alpha, \beta$ and hence the integral representation of $g'(\beta)$ is convergent and immediately shows it to be positive.  
Similarly $g'(\alpha)$ is {\em negative}. Note that this proof does not require the support to consists of a single band as long as we understand $\beta  = \sup {\rm supp}\, \rho_{\text{min}}$ and $\alpha  = \inf {\rm supp} \,\rho_{\text{min}}$. 
\end{proof}
Lemma \ref{acpositive} implies that there is no loss of generality in studying only the case $a>0$ since there is always a  positive critical $a_c$ (and a negative one); the negative case ($a<0$) is equivalent to the positive case by replacing $a\mapsto -a$, $V(z)\mapsto V(-z)$.

\begin{lemma}\label{lemmaP3}
The critical point structure of $\Re\, P_3(z)$ is: 
\begin{itemize}
\item For $a > a_c$, $\Re\, P_3(z)$ is strictly increasing on $\mathbb{R}\setminus[\alpha,\beta]$;

\item For $a=a_c$, $\Re\, P_3(z)$ is strictly increasing on $\mathbb{R}\setminus[\alpha,\beta]$ and $\Re\,P_3'(\beta)=0$;

\item For $0<a < a_c$, $\Re\, P_3(z)$ has unique local minimum on $\mathbb{R}\setminus(\alpha,\beta)$. This minimum occurs at a point $b^\star\in (\beta,\infty)$.

\end{itemize}
\end{lemma}
\begin{proof}
From the representation \eqref{g-form} of $g$ one sees immediately that 
\be
g''(z) =- \int_\R \frac 1{(z-s)^2} \rho_{\text{min}}(s) \d s\label{11}
\ee
which shows clearly that for $z\in \R \setminus {\rm supp\,} \rho_{\text{min}}$ the real part of $g$ is concave downward. Thus $\Re g'(z)$ is decreasing  in $\R \setminus {\rm supp\,} \rho_{\text{min}}$; moreover, from 
\be
g'(z) = \int_\R \frac 1{(z-s)} \rho_{\text{min}}(s) \d s \label{12}
\ee
we see that $\Re g'$  is negative  for $z<\inf {\rm supp\,} \rho_{\text{min}} = \alpha$ and positive  for  $z>\sup {\rm supp\,} \rho_{\text{min}}=\beta$.

 From the definition we see that $P_3'(z) = a - g'(z)$ and hence we infer:
\begin{itemize}

\item 
For $a > a_c=g'(\beta)>0$,  $ {\Re}\,(P_3'(z))=a-g'(z)>g'(\beta)-g'(z)$ is positive on $[\beta, \infty)$ therefore $\Re\,P_3$ is strictly increasing.  On the other hand  $a- g'(z)$ is clearly positive on  $(-\infty, \alpha]$ because $a>0$ and $- g'>0$ from \eqref{12}; 

\item For $a=a_c$, $ {\Re}\,P_3'(\beta) = 0$ and   $ {\Re}\, P_3'(z)$ is a monotonically increasing positive function on $(\beta, \infty)$, and a monotonically increasing positive function on $(-\infty, \alpha]$. Therefore there is a single critical point of $ {\Re}\, P_3(z)$ at $z=\beta$.  As $ {\Re}\, P_3''(z) = - {\Re}\, g''(z) > 0$ this must be a minimum;

\item For $0<a < a_c$,  $ {\Re}\, P_3'(\beta) < 0$ and  $P_3'(z)\to a > 0$ for $z \to \infty$; moreover  $ {\Re}\, P_3'(z)$ is a monotonically increasing function on $[\beta , \infty)$, and a (monotonically increasing) positive function on $(-\infty, \alpha]$.  
Since $P_3'(\beta)<0$ there must be a unique point $b^\star>\beta$ where $P_3'(b^\star)=0$.   As $ {\Re}\, P_3''(z) = - {\Re}\, g''(z) > 0 $ this must be the local minimum (or, equivalently, the global minimum on $(\beta,\infty)$).
\end{itemize}
 \end{proof}
We can now define four regimes:  supercritical, subcritical, critical, and 
jumping outliers.  We define the subcritical and critical regimes first.

\begin{defn}\label{subcritical}
The matrix model specified by \eqref{prob-measure} is in the {\bf subcritical regime} if $a < a_c$ and 
$P_2(x)<P_3(b^\star)$ for all $x\geq b^\star$.
\end{defn}

\begin{defn}\label{critical}
The matrix model specified by \eqref{prob-measure} is in the {\bf critical regime} if $a=a_c$ and $P_2(x)<P_2(\beta)$ for all $x>\beta$. 
\end{defn}

Now the supercritical regime can be efficiently defined as the remaining cases, except the small---codimension one---cases that we distinguish as the ``jumping outlier regime."
\begin{defn}
\label{supercritical}
The model is in the {\bf supercritical regime} if $P_2$ has a unique point of  global maximum on $\{x > \max\{\beta, b^\star\}\}$ at a point $x=a^\star\in \R$ and any of the three conditions below is satisfied:
\begin{itemize}
\item $a>a_c$.
\item $a=a_c$ and  $P_2(\beta)<P_2(x)$ for some $x>\beta$.
\item $0<a< a_c$ and $P_3(b^{\star})<P_2(x)$ for some $x>b^{\star}$.
\end{itemize}
Note that $a^\star$ is always greater than $\beta$ and $b^\star$.

If the global maximum of $P_2$ on $(\max\{\beta,b^\star\}, \infty)$ is  attained at several distinct points then we will say that we are in the {\bf jumping outlier regime} that also includes the following remaining case. 
\begin{itemize}
\item $0<a<a_c$ and  $P_2(x)=P_3(b^\star)$ for some $x> b^\star$. (The case $x=b^\star$ cannot occur for {\bf regular} $V$.)
\end{itemize}
\end{defn}


In the present paper we consider the supercritical and subcritical regimes.  The critical 
and jumping outlier regimes will be considered elsewhere 
\cite{Bertola:2010}.

The definition of the  supercritical regime is complicated and the reader may wonder whether the above definitions ever hold in actual examples. 
It is however not difficult to engineer a situation where they do occur, explained in the following example

\begin{example} [Second and third bullet in Definition \ref{supercritical}]
Consider a potential $V$ such that a new spectral band (i.e. interval of support of $\rho_{\text{min}}$) is about to emerge. Then $P_1$ has a local maximum $-E$  at $x_0>\beta$ outside of the main band which is slightly negative but small in absolute value. It is simple to arrange examples where $E$ is arbitrarily small. 
Since $a=a_c$ we have $\Re\, P_3'(\beta)=0$ (see \eqref{110}) and since $\Re\,P_3$ is concave upwards, it must be increasing for $x>\beta$. On the other hand $P_2 = P_1+P_3$ and thus 
\be
\Re\left(P_2(x_0)- P_2(\beta)\right) = -E + \Re(P_3(x_0)-P_3(\beta)).
\ee
Since $E$ can be chosen arbitrarily small, and since $\Re(P_3(x_0)-P_3(\beta))>0$, we see that necessarily we can have the situation described in the second bullet. 
By a continuity argument on $a<a_c$,  this also provides an example for the third bullet since  $P_3(b^\star)<P_3(\beta)\approx P_2(\beta)< P_2(x_0)$.   
\end{example}

Though we do not consider in this paper, one can also create an example in the jumping outlier regime.

We show in Proposition \ref{suffcond} that if $V$ is convex, then we are either in the super or subcritical depending on $a>a_c$ or $0\leq a <a_c$, respectively; in particular, in this case, the situation described in the second and third bullet points of Definition \ref{supercritical} cannot occur.

\bp
\label{suffcond}
Suppose that $V(z)$ satisfies Assumptions \ref{assumptionAV} and in addition it is convex ($V''>0$). Then 
\begin{enumerate}
\item[{\bf (i)}] for $a>a_c$ the model is supercritical and the maximum of $P_2$ at $a^\star$ is nondegenerate;
\item[{\bf (ii)}] for $0<a<a_c$ the model is subcritical;
\item[{\bf (iii)}] there is no jumping outlier regime.
\end{enumerate}
\ep
\begin{proof}
It is known that the convexity of $V$ is a sufficient condition for the support of the equilibrium measure to be a single band $[\alpha, \beta]$. 
Moreover from Lemma \ref{acpositive} we see $V'(\beta)=2g'(\beta)>0>V'(\alpha)=2g'(\alpha)$.
 Note that $P_j$ are all real-valued in $[\beta,\infty)$ (the cut of the logarithm  runs in $(-\infty, \beta]$).\\
{\bf (i)} We then observe that 
\be
P_2'' = -V'' + g''<0 
\ee
since both $-V$ (by assumption) and $g$ (by \eqref{11}) are concave downward. $P_2$ may have at most a single global (nondegenerate) maximum in $[\beta, \infty)$ because $P_2'(\beta) = a - \frac 1 2 V'(\beta) = a-a_c>0$. This also proves  ({\bf iii}). \\
{\bf (ii)} If $0<a<a_c$, $P_2$ strictly decreases on $[\beta,\infty)$ because $P_2'(\beta) = a - \frac 1 2 V'(\beta) = a-a_c<0$.; 
Also we have $P_2(b^\star)=P_1(b^\star)+P_3(b^\star)<P_3(b^\star)$.  Therefore $P_2(x)<P_3(b^\star)$ for all $x\geq b^\star$ as in Definition \ref{subcritical}.
\end{proof}


It should be noted here that, contrary to the work done for $V = z^2/2$, the position of $a$ relative to $a_c$ is not sufficient (for general $V$) to define the critical and subcritical regimes.  If, however, $V$ is  convex (for example an even monomial with positive coefficient) then by Proposition \ref{suffcond} the position of $a$ relative to $a_c$ determines the supercritical/subcritical regime as in \cite{Bleher:2004b, Aptekarev:2005, Bleher:2007, Adler:2009a, Peche:2006}.
The secondary conditions in Definition \ref{supercritical} are dealing with whether the Lagrange multiplier $\ell_2$ in the effective potentials $P_j$ can be chosen such that the off diagonal entries of the jump matrices for the deformed Riemann-Hilbert problems that we will construct in Sections 2, 3, and 4, decay to zero at an exponential rate. 
The problem of finding necessary and sufficient conditions on $V(x)$ and $a$ for the matrix model to be in the supercritical/subcritical regime is quite difficult, as much as it is difficult to find necessary and sufficient conditions for $V(x)$ to be a single-band potential.   
{We now specify the constant $l_2$:
\begin{defn}The constant $l_2$ in Definition \ref{P2-nolog} will be chosen as follows:
\label{defell2}
\begin{itemize}
\item In the supercritical case, the constant $l_2$ is chosen so that 
the unique global maximum of $P_2(z)$ on $(\beta,\infty)$ is zero (i.e. $\Re P_2(a^\star)=0$).
\item  In the subcritical case, the constant 
\eq
l_3:=-l_1+l_2
\endeq
is chosen so that $P_3(b^{\star})=0$.
\end{itemize}
\end{defn}}

\subsection{The kernel and its connection to multiple orthogonal polynomials}

Let $\rho_m(\lambda_1,\dots,\lambda_m)$ be the probability density that
the $n\times n$ matrix ${\bf M}$ chosen using \eqref{prob-measure} has 
eigenvalues
$\{\lambda_1,\dots,\lambda_m\}$ (here $m\leq n$).  Then, the $m$-point correlation function is
$R_m(\lambda_1,\dots,\lambda_m):=\frac{n!}{(n-m)!}\rho_m(\lambda_1,\dots,\lambda_m)$.  Br\'ezin and Hikami
\cite{Brezin:1996, Brezin:1997a, Brezin:1997b, Brezin:1998} showed that,
in the Gaussian case, the $m$-point correlation functions can
all be expressed in terms of a single kernel $K(x,y)$:
\eq
R_m(\lambda_1,\dots,\lambda_m) = \det(K(\lambda_i,\lambda_j))_{i,j=1,\dots,n}.
\endeq
Zinn-Justin \cite{Zinn-Justin:1997,Zinn-Justin:1998} extended this result to
the case of more general $V({\bf M})$.  Bleher and Kuijlaars \cite{Bleher:2004a}
rewrote the kernel in terms of {\it multiple orthogonal polynomials}, a
significant result because it allows one to analyze the asymptotic behavior of
these polynomials via the associated Riemann-Hilbert problem.

This approach was followed by Aptekarev, Bleher, and
Kuijlaars \cite{Bleher:2004b, Aptekarev:2005, Bleher:2007} in the Gaussian
case when the matrix ${\bf A}$ has two eigenvalues $\pm a$, each of
multiplicity
$n/2$.  When $a$ is sufficiently large the eigenvalues of ${\bf M}$
accumulate on two disjoint intervals (the supercritical case).  As $a$ 
decreases the two bands collide (the critical case).  Below this 
critical value of $a$, the eigenvalues accumulate on a single interval 
(the subcritical case).  Related behavior also appears in the 
theory of nonintersecting one-dimensional Brownian motions;  see, for 
instance, Adler, Orantin, and van Moerbeke \cite{Adler:2009b} for the 
critical case.

In general, the existence and number of bands on which eigenvalues accumulate
for large-rank sources for general $V({\bf M})$ is a complicated problem.  
For more on this question
see McLaughlin \cite{McLaughlin:2007} in which the quartic case
$V({\bf M}) = {\bf M}^4/4$ is worked out.  Bleher, Delvaux, and Kuijlaars 
\cite{Bleher:2010} have studied the external source problem with 
two eigenvalues of equal multiplicity and where $V({\bf M})$ 
is a sum of even-degree monomials with positive coefficients.
The external source with a finite number of different eigenvalues with various multiplicity 
for supercritical case has been considered in \cite{Delvaux:2010}.

The starting point of our analysis is the Riemann-Hilbert problem 
associated to the multiple
orthogonal polynomials.  Suppose ${\bf Y}(z)$ is a $3\times 3$
matrix-valued function of the complex variable $z$ satisfying
\eq
\label{rhp}
\begin{cases}
{\bf Y}(z) \text{ is analytic for } z\notin\mathbb{R}, \\
{\bf Y}_+(x) = {\bf Y}_-(x)\bpm 1 & e^{-nV(x)} & e^{-n(V(x)-ax)} \\ 0 & 1 & 0 \\ 0 & 0 & 1 \epm \text{ for } x\in\mathbb{R}, \\
{\bf Y}(z) = \left({\bf I}+O\left(\frac{1}{z}\right)\right)\bpm z^n & 0 & 0 \\ 0 & z^{-(n-r)} & 0 \\ 0 & 0 & z^{-r} \epm \text{ as } z\to\infty.
\end{cases}
\endeq
Here ${\bf Y}_\pm(x) := \lim_{\varepsilon\to 0}{\bf Y}(x\pm i\varepsilon)$ denote the
non-tangential limits of ${\bf Y}(z)$ as $z$ approaches the real axis from the
upper and lower half-planes.  Whenever posing a Riemann-Hilbert problem
we assume (unless otherwise stated) that the solution has continuous boundary
values along the jump contour when approached from either side.
Under our assumption (iv) in Section \ref{assumptions}, the unique 
solution ${\bf Y}(z)$ can be written
explicitly in terms of multiple orthogonal polynomials of the second
kind (see \cite{Bleher:2004b}, Section 2).
 In the case of two distinct  
eigenvalues $a$ and $0$, 
the kernel $K_n(x,y)$ may be written in terms of the function ${\bf Y}(z)$ as
\eq
\label{mop-kernel}
K_n(x,y) = \frac{e^{-\frac{1}{2}n(V(x)+V(y))}}{2\pi i(x-y)} \bpm 0 & 1 & e^{nay} \epm {\bf Y}(y)^{-1}{\bf Y}(x)\bpm 1 \\ 0 \\ 0 \epm.
\endeq
In the technical analysis of this Riemann-Hilbert problem we use and improve 
certain
ideas introduced by Bertola and Lee \cite{Bertola:2009a} to study the
first finitely many eigenvalues in the birth of a new spectral band for
the random Hermitian matrix model without source.

We note here that Baik \cite{Baik:2008} has recently expressed the
kernel $K_n(x,y)$ in terms of the standard (not multiple) orthogonal
polynomials.   
This offers an alternative method for approaching the
problem we consider here.
Based on this approach, Baik and Dong \cite{private} have obtained the universality result similar to ours, for the case of finite $r$ but possibly for non-degenerate eigenvalues, i.e. ${\bf A}=\text{diag}(a_1,a_2,\cdots,a_r,0,\cdots,0)$.
Since the rank of the matrices involved in
this alternative formulation grow with $r$, analyzing the Riemann-Hilbert
problem \eqref{rhp} seems more feasible if $r$ is allowed to grow sublinearly
with $n$.

\subsection{Assumptions on \texorpdfstring{${\bf A}$ and ${\bf V(z)}$}{AV} and results}
\label{assumptions}

First we gather the assumptions we will make in the rest of the paper.
\begin{assumption}
\label{assumptionAV}
We make the following requirements 
\begin{itemize}
\item[(i)]  $a>0$.
\item[(ii)]  ${\bf A}$ is a \emph{small-rank} external source of the form 
\eqref{A} with either $r$ a fixed positive integer or $r=\mathcal O(n^\gamma)$, with $0\leq \gamma<1$.  When $r$ 
is fixed we say $\gamma=0$.
\item[(iii)]  $V(z)$ is real-analytic.
\item[(iv)] $\ds\lim_{|z|\to\infty}\frac{V(z)}{\log(1+z^2)} = \infty \quad \text{and} \quad \lim_{|z|\to\infty}\frac{V(z)-az}{\log(1+z^2)} = \infty.$
\item[(v)] $V(z)$ is a {\em single-band} potential  (for example it can be convex).
\item[(vi)]  The density of the equilibrium measure of $V(z)$ has square root decay at its 
two endpoints (i.e. it is {\em regular} in the sense of \cite{Deift:1999a}).
\item[(vii)] For the supercritical regime, $P_2(z)$ behaves quadratically 
near $a^{\star}$.  Specifically, 
\eq
\label{P2-near-a*}
-V'(z) + a + g'(z) = -c(z-a^{\star}) + O( (z-a^{\star})^2) \text{ as } z\to a^{\star}
\endeq
for some constant $c>0$.  
\end{itemize}
\end{assumption}
Assumption (i) is for convenience, as the case when $a<0$ is equivalent by 
sending $a\to-a$ and $V(z)\to V(-z)$.  Regarding assumption (ii), 
in the general case when ${\bf A}$ has $m>2$ distinct eigenvalues the 
kernel can be written in terms of multiple orthogonal polynomials associated 
to an $(m+1)\times (m+1)$ Riemann-Hilbert problem, which is beyond the scope of 
this paper.  Assumption (iii) allows us 
to use the nonlinear steepest-descent method for Riemann-Hilbert problems, 
while (iv) guarantees the existence of the multiple orthogonal polynomials 
needed to ensure the Riemann-Hilbert problem has a solution.  

Assumption 
(v) avoids the necessity of using Riemann-theta functions for the solution 
of the outer model Riemann-Hilbert problem.  We expect similar results to 
hold generically in the multi-band case.  

Both (vi) and (vii) are 
genericity assumptions.  Assumption (vi) allows us to use Airy parametrices near the 
band endpoints.  Assumption (vii) produces Hermite (or Gaussian) behavior 
of the outlying zeros.  
Our results are computations of the large $n$ behavior of the kernel function 
\eqref{mop-kernel}.  We explicitly compute the kernel in a neighborhood of $a^{\star}$ for the supercritical regime and in a neighborhood of $b^{\star}$ for the subcritical regime. In the remaining portions of the complex plane, our result is that the kernel function converges to the kernel for the classical orthogonal polynomial problem with respect to $V(x)$.  In particular, our results include that, away from the $a^{\star}$ and $b^\star$,  the standard universality classes apply (i.e. a sine kernel in the bulk of the spectrum, and Airy kernels at the edges).  
\begin{shaded}
\begin{thm}\label{theorem-super-kernel}
Suppose $V(z)$ and $a$ satisfy assumptions (i)--(vii) and definition 
\ref{supercritical} of the supercritical regime.  
Let $\zeta_x$ and $\zeta_y$ be the local coordinates 
corresponding to $x$ and $y$ near $a^{\star}$ as defined in 
\eqref{super-zeta}.  
Uniformly for $\zeta_x, \zeta_y$ in compact sets we have the following asymptotics for $r= C n^\gamma$ for some $C>0$, $0\leq\gamma<1$.
\eq
\label{kernel-thm-form}
K_n(x(\zeta_x),y(\zeta_y)) = e^{-\frac{n}{2}P_3(x)+\frac{n}{2}P_3(y)}\frac{\sqrt{f''(0)}}{k_{r-1}^{(r)}}\kappa^{-1/2} \left(K_r^\text{GUE}(\zeta_x,\zeta_y)+\mathcal{O}(n^{-(1-\gamma)/2}) \right),
\endeq
where $P_3(x)$ is given by \eqref{P3-nolog}, $f(z;\kappa)$ is defined in \eqref{super-def-of-f}, $k_{r-1}^{(r)}$ is 
defined by \eqref{def-of-kr}, $\kappa=r/n$, and 
\eq
K_r^\text{GUE}(\zeta_x,\zeta_y):=\frac{H_r^{(r)}(\zeta_x-\zeta_0)H_{r-1}^{(r)}(\zeta_y-\zeta_0)-H_{r-1}^{(r)}(\zeta_x-\zeta_0)H_r^{(r)}(\zeta_y-\zeta_0)}{\zeta_x-\zeta_y}e^{-\frac{r}{4}(\zeta_x-\zeta_0)^2-\frac{r}{4}(\zeta_y-\zeta_0)^2}
\endeq
is the kernel for $r$ eigenvalues of the Gaussian Unitary Ensemble of scale 
$r$ centered at $\zeta_0$, which is defined by the change of variables 
\eqref{super-P2-zeta}.  Here $H_i^{(r)}(\zeta)$ are the rescaled monic 
Hermite polynomials satisfying the orthogonality condition \eqref{def-of-kr}.
\end{thm}
\end{shaded}

The presence of  $
\exp(-n P_3(x) / 2 ) 
$ in (\ref{kernel-thm-form}) 
 does not affect spectral properties of the kernel (because it amounts to a conjugation of the kernel by a diagonal operator) and therefore the implication is that asymptotically the eigenvalues near $a^{\star}$ are equivalent to those of a scaled $r\times r$ GUE problem; if $V(x)$ is a quadratic potential this  agrees with the results of \cite{Peche:2006}.
\begin{remark}
If the critical point of $P_2$ at $a^\star$ is more degenerate, $P_2(z) = \mathcal O((z-a^\star)^{2k})$, then one may follow similar steps as in \cite{Bertola:2009a} and \cite{Bertola:2009} and conclude that the relevant statistics of the outliers are determined by the kernel of a unitary ensemble with potential given by a polynomial of degree $2k$ instead of a Gaussian, namely, obtained from  a deformation of the Freud orthogonal polynomials.
\end{remark}
 
Theorem \ref{theorem-super-kernel} shows that, as expected, the equilibrium 
measure has no mass near $a^{\star}$.  We have 
the mean density of states 
\eq
\rho_n(x(\zeta_x) ) = \lim_{\zeta_y \to \zeta_x} \frac{1}{n} K_n( x(\zeta_x), y(\zeta_y) ) = \frac{\sqrt{f''(0)}}{k_{r-1}^{(r)}} \kappa^{1/2} \left( \rho_r^{(r)}( \zeta_x ) + \mathcal{O}(n^{-(1-\gamma)/2} r^{-1} )  \right)
\endeq
where $\rho_r^{(r)}(\zeta)$ is the mean density of eigenvalues for the $r \times r$ GUE ensemble.  If $r$ is fixed, this quantity is $\mathcal{O}(n^{-1/2})$ for large $n$, and if $r= n^\gamma$, using that $\rho_r^{(r)}(\zeta_x) \to \frac{1}{2\pi} \sqrt{ 4 - \zeta_x^2} $ as $r \to \infty$, we find 
\begin{equation}
\lim_{n\to \infty} \rho_n(x(\zeta_x)) = \frac{\sqrt{f''(0)}}{2 \pi k_{r-1}^{(r)}} \kappa^{1/2} \sqrt{ 4 - \zeta_x^2} \,.
\end{equation}
In either case, our conclusion is that for large $n$ the mean density of eigenvalues is asymptotically small (of order $\kappa^{1/2}$) in the neighborhood of $a^{\star}$ chosen in the theorem.   See Chapters 5 and 6 of \cite{Deift:1998-book}, Theorem 1.1 of \cite{Bleher:2004b}, and Theorem 8.1 of \cite{McLaughlin:2007} for similar results regarding the derivation of the asymptotic mean distribution of eigenvalues from the kernel.

\begin{shaded}
\begin{thm}
\label{theorem-sub-kernel}
Suppose the pair $(V(z),a)$ satisfies Definition \ref{subcritical} of 
the subcritical regime.  
There exists a closed disk of fixed radius centered at $b^{\star}$ such that, 
for $x$ and 
$y$ in this disk, for large $n$, and for $r=C n^\gamma$ for some $C>0$ and $0\leq \gamma<1$,  there is a $c>0$ such that 
\eq
K_n(x,y) = \mathcal{O}(n^{-(1-\gamma)/2} e^{-cn}).
\endeq
\end{thm}
\end{shaded}

\bigskip
\noindent
{\bf Acknowledgments.}  The authors would like to thank Jinho Baik,
Ken McLaughlin, Sandrine P\'ech\'e and Dong Wang for several illuminating discussions.  We thank Baik and Wang for sharing their unpublished results. 
M. Bertola 
was supported by NSERC.  R. Buckingham was 
supported by the Taft Research Foundation.  V. Pierce was supported by 
NSF grant DMS-0806219. 

\section{The supercritical regime}
\subsection{Modified equilibrium problem}
\label{supercrit}

In this section we will use a positive integer $K$; the general statements are valid for any $K$ but we will choose (for future use) 
\eq
\label{super-k-def}
K\geq \max \le\{\frac{3\gamma-1}{1-\gamma},0\ri\}
\endeq
where, we recall, $\gamma$ is the exponent of growth of $r=Cn^\gamma$ for some $C$ and $0\leq \gamma<1$.
%

We will need to build a {\em perturbation} of the equilibrium problem that leads to the definition of $g(z)$; we will denote by $\mathfrak g$ the resulting $g$--function of this perturbation scheme.
The construction, rather involved, will be broken down in steps.

The unperturbed equilibrium measure is supported on the single interval $[\alpha, \beta]$ (by assumption) with external field $V(z)$. Recall that $a^\star$ is lying {\em outside} of $[\alpha,\beta]$ and we fix a compact interval $J$ containing $[\alpha,\beta]$ in its interior and $a^\star\not \in J$.
\bp
\label{propdeform}
For any $K\in \mathbb N$ there is a neighborhood of the origin in $(\kappa,\vec \delta) \in \C^{1+K}$  such that the equilibrium measure $\widetilde \sigma(x)\d x$ of total mass $1-\kappa$ for the external field 
\be
\widetilde V(z):=V(z)  -  \delta V(z), \ \ \ \delta V(z):= \kappa \ln (z-a^\star ) + \kappa \sum_{j=1}^K \frac {\delta_j}{2(z-a^\star)^j}  
\label{deformV}
\ee
is supported on a single interval $[\alpha(\kappa,\vec \delta), \beta (\kappa,\vec \delta)]$ still contained in the interior of $J$: the endpoints $\alpha(\kappa,\vec \delta), \beta (\kappa,\vec \delta)$ are analytic functions of the specified variables.

Furthermore the (normalized) $g$--function of this problem 
\be
\mathfrak g(z):= \int \ln(z-w) \frac{\widetilde  \sigma(w)}{1-\kappa}\d w
\ee
converges uniformly over closed subsets not containing $[\alpha,\beta]$ to the unperturbed $g$--function.
\ep
\begin{remark}
In this proposition we treat the deformation parameters $\kappa, \vec \delta$ as independent from each other; later on, in Proposition \ref{propdeltas}, they will be uniquely determined in terms of the sole parameter $\kappa$. 
\end{remark}
\begin{proof}
It is well known (see, for example, \cite{Deift:1998a}) that 
\begin{itemize}
\item $\mathfrak{g}(z)$ is analytic for $z\notin (-\infty, \widetilde \beta] $, where $\widetilde \beta = \sup\ \text{supp}(\widetilde \sigma)$;
\item $\mathfrak{g}(z)$ has continuous boundary values and satisfies 
\eq
\label{super-g-rhp2}
\begin{split}
& \mathfrak{g}_+(z) - \mathfrak{g}_-(z) = 2 \pi i, \; z \in (-\infty, \widetilde \alpha), \\
& (1-\kappa) \left( \mathfrak{g}_+(z; \kappa) + \mathfrak{g}_-(z; \kappa) \right) = V(z) - \kappa \log(z - a^{\star}) -\sum_{j=1}^{2k}\frac{\delta_j}{2(z-a^{\star})^j}  - \ell_1, \; z \in (\widetilde \alpha,\widetilde  \beta)
\end{split}
\endeq
where $\widetilde \alpha = \inf\ \text{supp}(\widetilde \sigma)$ 
with the real axis oriented left to right;
\item $\mathfrak{g}(z; \kappa) = \log(z) + \mathcal{O}\left(\ds\frac{1}{z}\right) $ as $z \to \infty$.
\end{itemize}

Vice versa, the $g$--function may be characterized by the scalar Riemann--Hilbert problem (\ref{super-g-rhp2}) with the additional requirement that $\Im \mathfrak g_+$ is a nondecreasing function.

To show the analytic dependence of $\mathfrak g$ on $\kappa,\vec \delta$ we proceed as follows:
define the function 
\eq
\label{super-R-def}
R(z):=((z-\widetilde \alpha)(z-\widetilde \beta))^{1/2},
\endeq
where the principal branch of the square root is chosen so 
$R(z;\kappa)=z+\mathcal{O}(1)$ as $z\to\infty$.  
Taking the derivative of \eqref{super-g-rhp2} with respect to $z$ and using the
Plemelj formula gives
\eq
\label{super-gprime}
{\mathfrak g}'(z) = \frac{R(z)}{2\pi i(1-\kappa)}\int_{\widetilde \alpha}^{\widetilde \beta} \frac{V'(s) - \kappa/(s-a^{\star}) + \sum_{j=1}^{2k}j\delta_j/(2(s-a^{\star})^{j+1}) }{(s-z)R_+(s)} ds
\endeq
where $R_+(s)$ refers to the limit in $s$ from the upper half-plane.  
The large-$z$ expansion of (\ref{super-gprime}) along with the condition 
\eq
\mathfrak{g}'(z; \kappa) = \frac{1}{z} + \mathcal{O}\left( \frac{1}{z^2}\right)
\endeq
gives two conditions on $\alpha(\kappa), \beta (\kappa)$:
\begin{equation} \label{super-Jacobian1}
\int_{\widetilde \alpha}^{\widetilde \beta} \frac{V'(s) - \kappa/(s-a^{\star}) + \sum_{j=1}^{2k}j\delta_j/(2(s-a^{\star})^{j+1}) }{ R_+(s)} ds = 0 
\end{equation}
and 
\begin{equation} \label{super-Jacobian2}
\frac{1}{2\pi i} \int_{\widetilde \alpha}^{\widetilde \beta} \frac{V'(s) - \kappa/(s-a^{\star}) + \sum_{j=1}^{2k}j\delta_j/(2(s-a^{\star})^{j+1}) }{ R_+(s) } s \,ds = 1-\kappa \,.
\end{equation}
These two equations uniquely determine $\widetilde \alpha, \widetilde \beta$ as analytic functions of the parameters by the implicit function theorem.

The inequality $\Im \mathfrak g_+'>0$ remains valid, using a continuity argument,  for suitably small values of $\kappa, \vec \delta$ because it is valid (with the strict inequality) for the unperturbed $g$--function (by our initial assumption).

Therefore the expression (\ref{super-gprime}) yields a {\em bona fide} $g$--function for the modified external field $\widetilde V$ in a neighborhood of $(\kappa, \vec \delta ) = (0,\vec 0)$.

The expression for $\mathfrak g$ may be obtained by integration; specifically 
\be
\mathfrak g(z) = \int_{\widetilde\alpha}^z \mathfrak g'(s)\d s - \ell_1
\ee
and $\ell_1$ is also determined by the requirement that $\mathfrak g(z) = \ln (z) + \mathcal O(z^{-1})$ (without the constant term). Explicitly 
\be
\ell_1 = \lim_{z\to\infty} \le(\int_{\widetilde\alpha}^z \mathfrak g'(s)\d s - \ln z\ri)
\ee
which expression shows that $\ell_1$ is also analytic in the parameters, given the already proven analyticity of $\mathfrak g'$.

The statement about the convergence follows easily by noticing that $\mathfrak g'$ converges to $g'(z)$ as desired (note that they both have behavior $1/z + \mathcal O(z^{-2})$). This is best seen from the integral representations and the already proven analytic dependence on the deformation parameters.
\end{proof}

In parallel with the definition of the functions $P_j$ we shall define 
\begin{shaded}
\begin{eqnarray}
\label{super-mathcalP1}
{\mathcal P}_1(z;\kappa) & := & -\widetilde V(z) + 2(1-\kappa){\mathfrak g}(z;\kappa)  + \ell_1, \\
\label{super-mathcalP2}
{\mathcal P}_2(z;\kappa) & := & -\widetilde V(z) + az +\delta V  + (1-\kappa){\mathfrak g}(z;\kappa) +  
{l_2}, \\
\label{super-mathcalP3}
{\mathcal P}_3(z;\kappa) & := & -{\mathcal P}_1(z;\kappa) + {\mathcal P}_2(z;\kappa)\\
\widetilde V(z)& : =&V(z)  -  \delta V(z), \ \ \ \delta V(z):= \kappa \ln (z-a^\star ) + \kappa \sum_{j=1}^K \frac {\delta_j}{2(z-a^\star)^j}  
\end{eqnarray}
\end{shaded}  
\noindent where $l_2$ as in Definition \ref{defell2} and  is independent of the deformation.

Using deformation and continuity arguments we have that (for a sufficiently small value of the deformation parameters $\kappa, \vec \delta$) the real part of $\mathcal P_2$ has a global maximum in a neighborhood of $a^\star$. The main tool in the analysis  of the supercritical case when $r$ grows shall be the next theorem.
\begin{thm}
\label{rhoprop}
There exists a conformal change of coordinate $\rho = \rho(z;\kappa,\vec \delta)$ { fixing $z=a^\star$} (i.e. $\rho(a^\star;\kappa,\vec \delta) \equiv 0$)  that depends analytically on the parameters $\kappa, \vec \delta$ such that 
\be
\mathcal P_2 (z):= -V(z) + a z + (1-\kappa)\mathfrak g + l_2 +2\kappa\ln(z-a^\star) + \sum_{j=1}^K \frac {\delta_j}{(z-a^\star)^j} 
\ee
can be written as 
\be
\mathcal P_2(z;\kappa,\vec \delta) = -\frac 1 2 (\rho- {\mathfrak a} )^2 + 2\kappa \ln \rho+\mathfrak b  +  \sum_{j=1}^{K} \frac 
{ \gamma_j}{\rho^j}\label{mainid}
\ee
where the parameters ${\mathfrak a} = {\mathfrak a} (\kappa,\vec \delta)$, $\mathfrak b  = \mathfrak b (\kappa,\vec\delta)$ and $\vec \gamma = \vec \gamma(\kappa;\vec \delta)$ are analytic functions of the indicated parameters.
Furthermore the Jacobian 
\be
\frac {\pa \vec \gamma}{\pa \vec \delta} 
\ee 
is nonsingular in a neighborhood of the origin (i.e. for $\kappa$ and $\vec\delta$ sufficiently small).
\end{thm}
\begin{proof}
To simplify the notation we set $a^\star =0$ (up to a translation this entails no loss of generality).
We can write $\mathcal P_2$ as 
\be
\mathcal P_2 = - f(z;\kappa,\delta) +2\kappa\ln(z) + \sum_{j=1}^K \frac {\delta_j}{z^j} \,.
\ee
By the definition of $a^\star$ (which is now translated to $0$), the function $f(z;\kappa,\vec \delta)$ has the property that 
\be
f(z; 0, \vec 0) = \frac  C2  z^2(1 + \mathcal O(z))\ , C>0
\ee
and hence 
\be
f(z; \kappa, \vec \delta) = \frac C2 z^2(1 + \mathcal O(z)) + \mathcal O(\kappa,\vec \delta). 
\ee
Let us fix any (smooth) curve in the parameter space $\kappa(t), \vec \delta(t)$ and denote  by a $\pa $ its tangent vector; we then must show the identity (\ref{mainid}) for $t$ near 0. We suppress the notation of the dependence on $\kappa, \vec \delta$ for brevity in what follows, with the understanding that $f(z), {\mathfrak a} , \gamma_j, \mathfrak b $ all depend on them.

In this part we only sketch the main idea, leaving a full proof for Appendix \ref{confmap-appendix}.
Let $\mathbb D(r)$ be the open disk of radius $r>0$ and let $\Omega_{1}$ be the Banach {\bf manifold} of {\bf univalent, analytic functions} $\rho:\mathbb D(r)\to \C$ which fix the origin $\rho(0)=0$; this is a closed Banach submanifold of all univalent analytic functions because the evaluation map is continuous.
   Define now 
\be
\mathcal M:= \Omega_1 \times \C^{K+1} = \le \{{\bf p} =  (\rho,{\mathfrak a} , \mathfrak b ,\vec \gamma),\ \ \rho\in \Omega_1, \ \ {\mathfrak a} ,\mathfrak b  \in \C \ri \}
\ee
which is naturally also an {\em infinite dimensional} Banach manifold.
We are going to show that the ordinary differential equation in $\pa$ that derives from (\ref{mainid}) is integrable on $\mathcal M$; taking the implicit differentiation of (\ref{mainid}) we obtain
\bea\nonumber
&
  -\pa  f(z) +2\pa  \kappa\ln(z) + \sum_{j=1}^K \frac { \pa  \delta_j }{z^j}  = \le({\mathfrak a}  - \rho + 2 \frac \kappa \rho    - \sum_{j=1}^{K} \frac {j  \delta_j}{\rho^{j+1}} \ri) \pa  \rho 
+\pa  \mathfrak b  
- \pa  {\mathfrak a}  (\rho-{\mathfrak a} ) + 2\pa  \kappa \ln \rho + \sum_{j=1}^K \frac { \pa  \gamma_j }{\rho^j}
\\\nonumber
&\displaystyle\Longrightarrow\quad
\pa  \rho = 
\frac{
-\pa  f(z) +2\pa  \kappa\ln\le( \frac {z}{\rho}\ri) + \sum_{j=1}^K \frac {\pa  \delta _j}{z^j} +\pa  \mathfrak b  +  \pa  {\mathfrak a}  (\rho-{\mathfrak a} ) 
 - \sum_{j=1}^K \frac { \pa  \gamma_j }{\rho^j}}{
{\mathfrak a}  - \rho +  \frac \kappa \rho  - \sum_{j=1}^{K} \frac {j \gamma_j}{\rho^{j+1}} 
 }\\
&\displaystyle \Longrightarrow\quad
  \pa  \rho = 
\rho\frac{
- \rho^{K} \pa  f(z) +2\rho^{K} \ln\le( \frac {z}{\rho}\ri)\pa  \kappa + \sum_{j=1}^K \frac {\rho^{K} \pa  \delta_j}{z^j} +\rho^{K} (\rho-{\mathfrak a} ) \pa  {\mathfrak a}  
 - \sum_{j=1}^K  \pa  \gamma_j \rho^{K-j} } {
{\mathfrak a} \rho^{K+1} - \rho^{K+2} + 2 \kappa \rho^K  - \sum_{j=1}^{K} j \gamma_j\rho^{K-j}
 }\label{49}.
\eea
Formula (\ref{49}) should be regarded as defining a vector field on $\mathcal M$,  and this flow together with  $\rho(z;0,\vec 0) = \sqrt{2f(z;0,\vec 0)}$ gives an initial value problem. To see this we have to remember that the tangent space to $\mathcal M$ consists of 
\be
T\mathcal M:= \Omega_0\times \C^{K+2}
\ee
where $\Omega_0$ stands for the Banach vector space of bounded analytic functions on $\mathbb D(r)$ (without the requirement of being univalent) mapping $0$ to $0$.
The denominator vanishes generically at $K+2$ values $\rho_j$; since  $\pa  \rho$ must be an analytic function, the numerator must vanish at the same points and  this yields a linear system for the $K+2$ values $\pa  \mathfrak b , \pa  {\mathfrak a} , \pa  \gamma_1,\dots, \pa  \gamma_K$. To see how this works  in more detail, let $\rho_j$ be the roots of the denominator in (\ref{49})
\be
 - \rho^{K+2} + {\mathfrak a} \rho^{K+1} + 2 \kappa \rho^K  - \sum_{j=1}^{K} j \gamma_j\rho^{K-j} =- \prod_{j=1}^{K+2} (\rho-\rho_j)
\label{52}\ .
\ee

For $\kappa, {\mathfrak a} , \vec \gamma$ sufficiently small all the roots $\rho_j$ belong to the disk $\mathbb D(r)$ where $\rho(z)$ is univalent and therefore there are corresponding values $z_1,\dots, z_{K+2}$. 

The linear system that determines $\pa{\mathfrak a} , \pa \vec \gamma$ is then 
\be
\rho_{\ell}^{K}\le(  - \pa  f(z_\ell) +2\ln\le( \frac {z_\ell }{\rho_\ell }\ri)\pa  \kappa\ri) + \sum_{j=1}^K \frac {\rho_\ell^{K}}{z_\ell ^j}  {\pa  \delta_j}+\rho_\ell ^{K}(\rho_\ell -{\mathfrak a} ) \pa  {\mathfrak a}  +\rho_{\ell}^{K} \pa \mathfrak b  + 
  \sum_{j=1}^K{\rho_\ell^{K-j}}{ \pa  \gamma_j }=0\ ,\ \ \ell=1,\dots, K + 2 .\label{linsys}
\ee
What we want to see is that this system determines $\pa {\mathfrak a}, \pa {\mathfrak b},\pa \vec \gamma$ as {\em analytic} functions of $\kappa, {\mathfrak a} , \vec \gamma$; to see this we observe that the coefficient matrix of the linear system (\ref {linsys}) is 
\bea
\begin{bmatrix}
\rho_1^{K}(\rho_1-{\mathfrak a} )  & \rho_1^K & \rho_1^{K-1} & \dots & 1\\
\rho_2^{K}(\rho_2-{\mathfrak a} ) & \rho_2^K & \rho_2^{K-1} & \dots & 1\\
\vdots&&&\\
\rho_{K+2}^{K}(\rho_{K+2}-{\mathfrak a} )   & \rho_{K+2}^K & \rho_{K+2}^{K-1} & \dots & 1
\end{bmatrix}\begin{bmatrix}
\pa {\mathfrak a} \\
\pa \mathfrak b \\
\pa \gamma_K\\
\vdots \\
\pa \gamma_1
\end{bmatrix} =  - \begin{bmatrix}
H(z_1)\\
H(z_2)\\
\vdots\\
H(z_{K+2})
\end{bmatrix}\ ,\\
H(z_{\ell}):= \rho_{\ell}^{K}\le( 2\ln\le( \frac {z_\ell }{\rho_\ell }\ri)\pa  \kappa-\pa  f(z_\ell) \ri) + \sum_{j=1}^K \frac {\rho_\ell^{K}}{z_\ell ^j}  {\pa  \delta_j}\,.
\eea
Solving this system by Cramer's rule yields $\pa{\mathfrak a} ,\pa \vec \gamma$ as {\bf symmetric} functions of the roots $\rho_\ell$; moreover it is seen that the determinant of the linear part is simply the Vandermonde determinant $\Delta(\vec \rho):= \prod_{j<\ell\leq K+2} (\rho_j-\rho_\ell)$ and since the determinants in the numerators of Cramer's formula will also vanish whenever two roots coincide, it follows that the ratio is actually {\em analytic} on the diagonals $\rho_\ell = \rho_k, \ \ell\neq k$.

We have thus proved that $\pa {\mathfrak a} , \pa \vec \gamma, \pa \mathfrak b $ are analytic symmetric functions of the $\rho_\ell$'s; it is well known that the ring of analytic symmetric functions is generated by the elementary symmetric polynomials in the roots, namely, the coefficients of the polynomial (\ref{52}). This means that  $\pa {\mathfrak a} , \pa \vec \gamma, \pa \mathfrak b $ are actually expressible in terms of analytic expressions of ${\mathfrak a} , \kappa, \vec \gamma$.

In order to complete the proof we should check that the vector field determined by (\ref{49}) is {\bf Lipshitz} with respect to the Banach norm of $T\mathcal M$; the check is rather straightforward but lengthy and a detailed analysis is deferred to App. \ref{confmap-appendix} in the simplified case $K=0$. After this, the existence and uniqueness of the integrated flow follows from standard theorems in Banach spaces.

\paragraph {\bf Jacobian at the origin.} To compute the Jacobian at the origin $(\kappa, \vec \delta)=(0,\vec 0)$ we have to set 
\be
\rho =  \sqrt {2f(z; 0,\vec 0)}\ ,\ \ {\mathfrak a} =\kappa = \delta_j = \gamma_j =\mathfrak b  =0\ .
\ee 
Taking now  $\pa_\ell $ to mean $\partial_{\delta_\ell}$ we find the equations 
\bea
\pa_\ell \rho =  \frac {-\pa_\ell f + \frac 1 {z^\ell } + \pa_\ell \mathfrak b   + \rho\pa_\ell {\mathfrak a} - \sum_{j=1}^K \frac {\pa_\ell \gamma_{j }}{\rho^j}}{-\rho}\ .
\label{50}
\eea
Since we want $\rho(0;\kappa,\vec \delta)\equiv 0$ we must impose that $\pa_\ell \rho$ in (\ref{50}) vanishes at least of order $z$ at $z=0$; this yields a linear system for the coefficients $\pa_\ell \mathfrak b ,\pa_\ell  {\mathfrak a}, \pa_\ell \gamma_{j}$ and in particular 
\bea
\le\{\begin{array}{ccl}
\pa_\ell \gamma_{j} (0,\vec 0) &=& 0\ , j> \ell\\
\pa_j \gamma_{j} (0,\vec 0) &=& 1 \\
\pa_\ell \gamma_{j}(0,\vec 0) &=& \star\ ,  j<\ell.
\end{array}\ri.\qquad
\frac {\pa \vec \gamma}{\pa  \vec \delta} =   \begin{bmatrix}
1& \star& \star& \dots\\
&1& \star & \dots\\
&&\ddots&\\
&&&1
\end{bmatrix}
\eea
with the $\star$ denoting some expression which is not relevant to us; the above Jacobian is then triangular with $1$ on the diagonal, and hence it is invertible in a neighborhood of $\kappa=0,\ \vec \delta =\vec 0$.
\end{proof}

\subsection{Determination of the \texorpdfstring{$\delta_j$'s}{deltaj}}

We now introduce the rescaled variable  ($\kappa>0$) 
\be
\zeta = \frac \rho {\sqrt\kappa} \ ,\ \ \ \zeta_0:= \frac {{\mathfrak a} }{\sqrt\kappa}\,.
\ee
Let 
\bea
\label{super-gH-def}
g_H(\zeta):= -\frac \zeta 4 \sqrt{\zeta^2 -4} + \ln (\zeta + \sqrt{\zeta^2-4})+ \frac {\zeta^2}2 + \frac {\ell_H} 2\\
\label{super-ellH}
\ell_H:=-1-2\log 2
\eea
be the $g$--function for the Gaussian Unitary Ensemble. It admits an asymptotic expansion of the form 
\be
g_H(\zeta):=\ln \zeta  + \sum_{\ell=1}^\infty \frac {c^{(0)}_\ell}{\zeta^{2\ell}}\ .
\ee
We define the constants $c_j^{(H)}$ by the requirement 
\be
\label{super-cH-def}
g_H(\zeta-\zeta_0) - \ln \zeta + \sum_{j}^{K} \frac {c_j^{(H)}}{\zeta^j} = \mathcal O(\zeta^{-K-1})\ , \ \ \zeta\to\infty.
\ee
It is easily verified that the $c_j^{(H)}$ are polynomials in $\zeta_0$.

\bp
\label{propdeltas}
The parameters $\vec \delta$ are uniquely determined as Puiseux series of $\sqrt{\kappa}$ by the requirement 
\bea
\mathcal P_2(z) = -\frac \kappa 2 (\zeta-\zeta_0)^2 + 2 \kappa \ln (\sqrt{\kappa} \zeta)  + \kappa \sum_{j=1}^K \frac {\kappa^{-\frac j 2} \gamma_j}{\zeta^j} + \mathfrak b  =  -\frac \kappa 2 (\zeta-\zeta_0)^2 + 2 \kappa \ln (\sqrt{\kappa} \zeta)  +\mathfrak b +2  \kappa \sum_{j=1}^K \frac {c^{(H)}_j}{\zeta^j}
\eea
Moreover we have 
\be
\zeta_0 = \mathcal O(\sqrt{\kappa}), \ \ \mathfrak b = \mathcal O(\kappa)\ ,\ \ \ \vec \delta = \mathcal O(\kappa^\frac 32).
\ee
\ep
\begin{proof}
Recall that $ \zeta_0$ depends on both $\kappa, \vec \delta$ analytically; although it is possible to give more detailed information about this dependence, it will not be necessary to the end of establishing the present proposition.

We need to solve the nonlinear system 
\bea
\kappa^{-\frac j2 }\gamma_j(\kappa,\vec \delta) = 2\kappa c_j^{(H)}( \zeta_0) \ ,\ \ \ j=1,\dots, K
\eea
for the unknowns $\vec \delta$. The local solvability of the system around $\vec \delta =\vec 0$ is guaranteed if we can guarantee the nonsingularity of the Jacobian matrix. But this  system can be rewritten 
\be
\vec \gamma - 2\kappa^D \vec c^{(H)} =0\ ,\ \ D:= {\rm diag}\le(\frac 32 , 2,\frac 5 2,  \dots \frac {K+2}2\ri)\label{64}.
\ee
Since $\vec c^{(H)}$ is analytic in $\vec \delta$ one promptly sees that the Jacobian is  
\be
J:= \frac {\pa \vec \gamma}{\pa \vec \delta} -2  \kappa^D \frac{\pa \vec c^{(H)}}{\pa \vec \delta}
\ee
and hence $\det J = 1 + \mathcal O(\kappa^{\frac 32})$. This guarantees that there is a polydisk $|\kappa|<C_1,\ \|\vec \delta\|<C_2$ (for suitable constants) where the system admits a solution  in Puiseux series (i.e. analytic in $\sqrt\kappa$).  It is also clear from \eqref{64} that $\vec \delta = \mathcal O(\kappa ^\frac 3 2)$. Thus, $\zeta_0(\kappa, \delta(\kappa))$ is still of order $\mathcal O(\sqrt\kappa)$ and $\mathfrak b(\kappa, \delta(\kappa))$ is still $\mathcal O(\kappa)$, since all these depend analytically on $\vec \delta$, which is not of higher order than $\kappa$.
\end{proof}

For future reference and definiteness we collect the result of the above discussion in the theorem below

\begin{shaded}
\begin{thm}
\label{confchange}
There exists a conformal change of coordinate $\zeta(z;\kappa)$ of the form 
\be
\zeta(z;\kappa)= \frac {\rho(z;\kappa)}{\sqrt{\kappa}} = \frac 1{\sqrt{\kappa}} C (z-a^\star) (1+\mathcal O(z-a^\star))\ , C>0
\label{super-zeta}
\ee
and a choice of $\vec \delta = \vec \delta(\kappa)$ for the deformed potential (\ref{deformV}) in Puiseux series of $\sqrt{\kappa}$ 
such that 
\be
\mathcal P_2(z; \kappa, \vec \delta(\kappa)) = - \frac \kappa 2 (\zeta - \zeta_0)^2 + 2 \kappa \ln (\sqrt \kappa \zeta) +\mathfrak b  + 2 \kappa \sum_{j=1}^K \frac {c^{(H)}}{\zeta^j}\ .
\label{mainidzeta}
\ee
The functions  $\zeta_0(\kappa), \beta(\kappa), \vec \delta(\kappa)$ admit a Puiseux expansion and are of orders 
\be
\zeta_0 = \mathcal O(\sqrt\kappa)\ ,\ \ \beta = \mathcal O(\kappa)\ ,\ \ \vec \delta = \mathcal O(\kappa).
\ee
The expressions $c^{(H)}_j$ are polynomials of degree $j$ in $\zeta_0$ determined by the formula \eqref{super-cH-def}.
\end{thm}
\end{shaded}


\subsection{Steepest descent analysis (supercritical case)}

We make the following change of variables from ${\bf Y}(z)$ to
${\bf W}(z)$:
\bea
\label{super-Y-to-W}
{\bf W}(z;\kappa)&\hspace{-40pt}:=\left(\begin{array}{ccc}e^{\frac{n}{2}\ell_1} & 0 & 0 \\0 & e^{-\frac{n}{2}\ell_1} & 0 \\0 & 0 & e^{\frac{n}{2}
(\ell_1-2l_2+
2\eta)
}\end{array}\right)
 {\bf Y}(z)\left(\begin{array}{ccc}e^{-\frac{n}{2}V} & 0 & 0 \\0 & e^{\frac{n}{2}V} & 0 \\0 & 0 & e^{\frac{n}{2}(V-2a z)}\end{array}\right)\times
 \\ &\times
\left.\begin{cases}
{\bf I}, &z\in \Omega_1\cup\Omega_6
\\ \left(\begin{array}{ccc}1 & 0 & 0 \\ 0 & 1 & -1 \\ 0 & 0 & 1 \end{array}\right) , &z\in \Omega_2 \cup \Omega_5
\\ \left(\begin{array}{ccc}1 & 0 & 0 \\ -1 & 1 & -1 \\0 & 0 & 1\end{array}\right) , &z\in \Omega_3
\\ \left(\begin{array}{ccc}1 & 0 & 0 \\ 1 & 1 & -1 \\0 & 0 & 1\end{array}\right) , &z\in \Omega_4
\end{cases}\right\}
\left(\begin{array}{ccc}
e^{-\frac{n}{2}{\mathcal P}_1} & 0 & 0 \\
0 & e^{\frac{n}{2}{\mathcal P}_1 } & 0 \\
0 & 0 & e^{\frac{n}{2}(2\mathcal P_2  - \mathcal P_1 -
2\eta ) }
\end{array}\right) e ^{-\frac n2 \delta V(z)} \\
  & 
  \delta V(z):= \kappa \ln (z-a^\star ) + \kappa \sum_{j=1}^K \frac {\delta_j}{2(z-a^\star)^j}    
\\
&\eta:= \kappa \ln \kappa  + \mathfrak b  + \kappa \ell_H\ ,\ \ \ (\ell_H:= -1-2\ln2).
\label{defeta}
\eea
The constant (in $z$) $\eta$ is chosen to carefully balance other constants later on in the study of the local parametrix; we recall that $\mathfrak b = \mathcal O(\kappa)$ has appeared in Theorem \ref{confchange}, which is also the source of the $\kappa \ln \kappa $ term. The constant $\ell_H$ was introduced in \eqref{super-ellH} and it is the Robin's constant for the equilibrium problem associated to the quadratic potential.
See Figure \ref{fig_biglensSuper} for a visual on the different regions $\Omega_j$'s.  The exact choice of the outer lenses 
is given below in the proof of Lemma \ref{P-nolog-lemma}(d).  
The inner lenses are chosen in the standard way for the $2\times 2$ 
Riemann-Hilbert problem for (non-multiple) orthogonal polynomials.  

The new matrix  ${\bf W}(z)$ satisfies a new Riemann Hilbert Problem  which can be directly evinced from the one for ${\bf Y}$ and is of the form  ${\bf W}_+(z)={\bf W}_-(z){\bf V^{(W)}}(z)$  with jumps 
\begin{equation} \label{super-Wjumps}
{\bf V^{(W)}}(z) = \left\{\begin{array}{ll}
\bpm 1 & e^ {n{\mathcal P}_1(z)} & e^{n({\mathcal P}_2(z)
-\eta)
} \\ 0 & 1 & 0 \\ 0 & 0 & 1 \epm, & z\in \partial\Omega_1\cap\partial\Omega_6,
\\\bpm 1&0&0\\0&1&-e^{n({\mathcal P}_3(z)-\eta)}\\0&0&1\epm, & z\in (\partial\Omega_1\cap\partial\Omega_2) \cup (\partial\Omega_5\cap\partial\Omega_6),
\\\bpm 1 & 0 & 0 \\ e^{-n{\mathcal P}_1(z)} & 1 & 0 \\ 0 & 0 & 1 \epm, & z\in (\partial\Omega_2\cap\partial\Omega_3)\cup(\partial\Omega_4\cap\partial\Omega_5),
\\\bpm 0&1&0\\-1&0&0\\0&0&1\epm, & z\in[\alpha,\beta],
\\\bpm 1& e^{n{\mathcal P}_1(z)}  &0\\0&1&0\\0&0&1\epm, & z\in\partial\Omega_2\cap\partial\Omega_5,
\end{array}\right.
\end{equation}
and the asymptotic conditions
\bea
\label{super-W-infinity}
&& {\bf W}(z)={\bf I}+{\cal O}\left(\displaystyle\frac{1}{z}\right), \quad  z\to\infty\\
\label{super-Wzero}
&& \begin{split}
{\bf W}(z) = \text{(analytic)}\bpm 
{ {\rm e}^{- n \delta V(z)}}
 & 0 & 0 \\ 0 & 1 & 0 \\ 0 & 0 & 
{ {\rm e}^{n \delta V(z)}}
  \epm& \\
\text{as }z\to a^{\star}.&
\end{split}
\eea
The orientation of the contours is given in Figure \ref{fig_biglensSuper}.
Here we have used the factorization
\eq
\bpm 1 & 1 & 1 \\ 0 & 1 & 0 \\ 0 & 0 & 1 \epm = \bpm 1 & 0 & 0 \\ 0 & 1 & -1 \\ 0 & 0 & 1 \epm  \bpm 1 & 0 & 0 \\ 1 & 1 & 0 \\ 0 & 0 & 1 \epm \bpm 0 & 1 & 0 \\ -1 & 0 & 0 \\ 0 & 0 & 1 \epm  \bpm 1 & 0 & 0 \\ 1 & 1 & 0 \\ 0 & 0 & 1 \epm  \bpm 1 & 0 & 0 \\ 0 & 1 & 1 \\ 0 & 0 & 1 \epm.
\endeq


\begin{figure}
\setlength{\unitlength}{2.7pt}
\begin{center}
\begin{picture}(100,50)(-50,-25)
\put(-10,0){\circle*{1}}
\put(-11.5,2){$\alpha$}
\put(10,0){\circle*{1}}
\put(9.5,2){$\beta$}
\put(28,0){\circle*{1}}
\put(27,2){$a^{\star}$}
\put(-50,0){\line(1,0){90}}
\put(-42,14){$\Omega_1$}
\put(-22,8.5){$\Omega_2$}
\put(-1.5,3){$\Omega_3$}
\put(-1.5,-4.5){$\Omega_4$}
\put(-22,-9.5){$\Omega_5$}
\put(-42,-15){$\Omega_6$}
\qbezier(-10,0)(0,16),(10,0)
\qbezier(-10,0)(0,-16),(10,0)
\qbezier(15,0)(15,25)(-10,25)
\qbezier(15,0)(15,-25)(-10,-25)
\qbezier(-35,0)(-35,25)(-10,25)
\qbezier(-35,0)(-35,-25)(-10,-25)
\thicklines
\put(-42,0){\vector(1,0){1}}
\put(-22,0){\vector(1,0){1}}
\put(0,0){\vector(1,0){1}}
\put(22,0){\vector(1,0){1}}
\put(35,0){\vector(1,0){1}}
\put(-10,-25){\vector(1,0){1}}
\put(-10,25){\vector(-1,0){1}}
\put(0,8){\vector(1,0){1}}
\put(0,-8){\vector(1,0){1}}
\end{picture}
\end{center}
\caption{The oriented jump contour $\Gamma$ for {\bf W}(z) and the regions $\Omega_i$ in the supercritical case.\label{fig_biglensSuper}}
\end{figure}
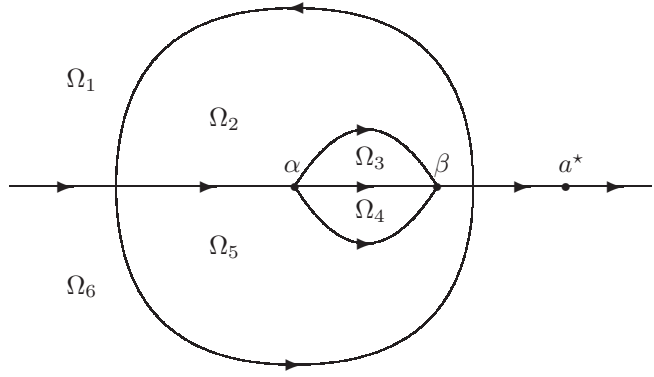

\subsubsection{The outer parametrix}

We will show below in Section \ref{super-error} that the jump matrices for ${\bf W}(z)$ decay uniformly to constant jump matrices
as $n\to\infty$ outside of small fixed neighborhoods of $\alpha$, $\beta$,
and $a^{\star}$.  These limiting jump matrices are the identity except on the
band $[\alpha,\beta]$.  We therefore define the outer parametrix
${\bf \Psi}(z)$ to be the solution to the following Riemann-Hilbert problem:
\eq \label{Vpsi}
{\bf \Psi}_+(z) = {\bf \Psi}_-(z)\bpm 0 & 1 & 0 \\ -1 & 0 & 0 \\ 0 & 0 & 1 \epm \text{ for } z\in(\alpha,\beta)\ , \qquad {\bf \Psi}(z)={\bf I}+{\cal O}\left(\frac{1}{z}\right).
\endeq
It is well known that the solution to this Riemann-Hilbert problem is
\eq
{\bf \Psi}(z) = {\bf U}^{-1}\bpm \ds \left(\frac{z-\beta}{z-\alpha}\right)^{-1/4} & 0 & 0 \\ 0 & \ds \left(\frac{z-\beta}{z-\alpha}\right)^{1/4} & 0 \\ 0 & 0 & 1 \epm {\bf U}, \quad {\bf U}:=\bpm  \frac{1}{2} &  \frac{i}{2} & 0 \\  -\frac{1}{2} &  \frac{i}{2} & 0 \\ 0 & 0 & 1 \epm,
\endeq
where
\eq
\lim_{z\to\infty}\left(\frac{z-\beta}{z-\alpha}\right)^{1/4} = 1
\endeq
and this function is cut along $[\alpha,\beta]$.
\subsubsection{The local parametrix near \texorpdfstring{$\bs{a^{\star}}$}{astar} }
Special attention is needed near the point $z=a^{\star}$ as the jump 
matrices do not decay
uniformly to the identity near this point.  
According to the definition of $a^\star$ as the point of maximum for $P_2$ and given that $\mathcal P_2$ is a deformation of $P_2$ we have
\eq
\label{super-def-of-f}
{\mathcal P}_2(z;\kappa) = - \overbrace{ \frac{C}{2}(z-a^{\star})^2(1+ \mathcal O(z-a^\star))  +\mathcal O(\delta) }^{:=f(z;\kappa)} + 2\kappa\log(z-a^{\star}) + \sum_{j=1}^{2k}\frac{\delta_j(\kappa)}{(z-a^{\star})^j} 
\endeq
where $C>0$ and the deformation $\mathcal O(\delta)$ is some analytic function of $z$ of the indicated order in $\kappa$.
  Let $\mathbb{D}_{a^{\star}}$ 
be a fixed-size 
circular disk centered at $a^{\star}$ chosen small enough so that
\eq
\label{condition-on-Da*}
\left|\Re\left[
\frac {f(z;\kappa)}2\right]\right|<|\Re\, P_1(z)|
\endeq
(recall that $\Re P_1>0$ outside of the support of the equilibrium measure)
inside the disk, and such that  the disk does not intersect the outer lenses.  
Orient $\partial\mathbb{D}_{a^{\star}}$ clockwise.  

We now apply Theorem \ref{confchange}: 
let the $k$ constants $c_j^{(H)}$, $j=1,\dots,k$ be specified by 
\eqref{super-cH-def}:
in the local scaling  coordinate $\zeta$  we have 
\be
\label{super-P2-zeta}
n\mathcal P_2 = -\frac r 2 (\zeta- \zeta_0)^2  + 2r \ln \zeta  + 2 r \sum_{j=1}^K \frac {c^{(H)}} {\zeta^j} + r \ln \kappa + n \mathfrak b.
\ee
where the scaling coordinate $\zeta$ has the following behavior on the boundary of the disk $\mathbb D_{a^\star}$ 
\eq
\label{super-zeta-order}
\zeta = \mathcal{O}(n^{(1-\gamma)/2}) \quad \text{when } z\in\partial\mathbb{D}_{a^{\star}}.
\endeq
We also recall (Theorem \ref{confchange}) that $\mathfrak b  =\mathcal O(\kappa)$ and hence $n\mathfrak b = \mathcal O(r)$.
This suggests the following definition for the model Riemann Hilbert Problem  of the local parametrix.
\begin{defn} The local parametrix within the disk $\mathbb D_{a^\star}$ shall be the unique solution ${\bf R}(z)$ to the  following model Riemann-Hilbert:
\bea
\label{R-RHP}
\begin{cases}
& {\bf R}_+(\zeta) = {\bf R}_-(\zeta)\bpm 1 & 0 & \zeta^{2r}\exp\left(-\frac{r}{2}(\zeta-\zeta_0)^2
 + 
 r\ell_H+2r\sum_{j=1}^K\frac{c_j^{(H)}}{\zeta^{j}}\right) \\ 0 & 1 & 0 \\ 0 & 0 & 1 \epm  
\\ & \phantom{{\bf R}_+(\zeta)} = {\bf R}_-(\zeta)\bpm 1 & 0 & e^{n(\mathcal P_2 - \eta)} \\ 0 & 1 & 0 \\ 0 & 0 & 1 \epm, \quad \zeta\in\mathbb{R},
\\ &{\bf R}(\zeta)={\bf I}+{\cal O}\left(\ds\frac{1}{\zeta}\right) \text{ as } \zeta\to\infty,
\\&{\bf R}(\zeta)=(\mbox{analytic})\bpm \zeta^{-r}\exp\left(-r\sum_{j=1}^K \frac{c_j^{(H)}}{\zeta^{j }}\right) & 0 & 0 \\ 0 & 1 & 0 \\ 0 & 0 & \zeta^r\exp\left(r\sum_{j=1}^K \frac{c_j^{(H)}}{\zeta^{j }}\right) \epm \text{ as } \zeta\to 0.
\end{cases}\\
\eta := \kappa\ln \kappa + \mathfrak b - \kappa\ell_H
\eea
\end{defn}
The behavior at $\zeta=0$ ($z=a^\star$) is dictated by \eqref{super-Wzero}.
We point out that the problem is essentially $2\times 2$; moreover it will be shown below that it is a slight modification of the Fokas-Its-Kitaev Riemann-Hilbert
problem for Hermite orthogonal polynomials (see \cite{Fokas:1992} and
\cite{Deift:1999b}, Section 3).  
\begin{prop}
\label{propHparametrix}
Let the rescaled Hermite polynomials
$H^{(r)}_i(\zeta)$ be the family of monic polynomials satisfying the
orthogonality condition
\begin{equation}
\label{def-of-kr}
\int_{-\infty}^\infty H^{(r)}_i(\zeta)H^{(r)}_j(\zeta) e^{-\frac r2 \zeta^2}d\zeta=r^{j-\frac 1 2} j! \sqrt{2\pi}\delta_{ij} =k^{(r)}_j\delta_{ij},
\end{equation}
where the $k^{(r)}_i$ are normalization constants.
Then the solution to \eqref{R-RHP} is
\eq
\label{super-R}
{\bf R}(\zeta) =  \exp\left(-\frac{r}{2}
\ell_H
{\bf\Lambda}_{13}\right) \,{\bf H}_{13}(\zeta) \,\zeta^{-r{\bf\Lambda}_{13}}\,\exp\left(\left(\frac{r}{2}
\ell_H
 - r\sum_{j=1}^K \frac{c_j^{(H)}}{\zeta^{j}}\right){\bf\Lambda}_{13}\right),
\endeq
where
\eq
\label{super-H}
{\bf H}_{13}(\zeta):=\bpm H^{(r)}_r(\zeta-\zeta_0) & 0 &\displaystyle \frac{1}{2\pi i}\int_{-\infty}^\infty \frac{H^{(r)}_r(s-\zeta_0) e^{-\frac{r}{2}s^2}}{s-\zeta}ds \\ 0 & 1& 0 \\ \displaystyle\frac{2\pi i}{-k^{(r)}_{r-1}}H^{(r)}_{r-1}(\zeta-\zeta_0)& 0 & \displaystyle \frac{-1}{k^{(r)}_{r-1}}\int_{-\infty}^\infty \frac{H^{(r)}_{r-1}(s-\zeta_0) e^{-\frac{r}{2}s^2}}{s-\zeta}ds \epm
\endeq
and 
\eq
{\bf\Lambda}_{13}:=\bpm 1 & 0 & 0 \\ 0 & 0 & 0 \\ 0 & 0 & -1 \epm.
\endeq
\end{prop}
The proof is a direct manipulation and it is left to the reader.\\

Now  the well known asymptotics of Hermite polynomials can be written as a {\bf joint} asymptotic expansion for $\zeta $ {\bf and} $r$ large as follows
\bea
{\bf H}_{13}(\zeta)=e^{\frac{r}{2}\ell_H {\bf\Lambda}_{13}} \left({\bf I}+{\cal O}\left(\frac{1}{r(|\zeta|+1)}\right)\right)
\mathbf U^{-1}
\left(\begin{array}{ccc} 
 \sqrt[4]{\frac {\zeta-\zeta_0 -2}  {\zeta-\zeta_0 +2}} &0&0  \\0&1&0\\
0&0&\sqrt[4]{\frac {\zeta-\zeta_0 +2}  {\zeta-\zeta_0  -2}}   \end{array}\right)
\mathbf U
 e^{r g_H(\zeta-\zeta_0){\bf\Lambda}_{13}-\frac{r}{2}\ell_H{\bf\Lambda}_{13}}\ ,
 \label{Hasympt}
  \\
\mathbf U:= \frac 1 {\sqrt{2}} \begin{bmatrix}
1 & 0 & i\\
0 & 1 & 0\\
-1 & 0 & i
\end{bmatrix}
\eea
and  $g_H(\zeta)$  is given by \eqref{super-gH-def}.
For large $\zeta$, we have the expansion of $g_H$ as in \eqref{super-cH-def}.
The error term of ${\cal O}(1/r(|\zeta|+1))$ in \eqref{Hasympt} 
is from the standard Airy parametrix used in the Riemann-Hilbert problem 
for Hermite orthogonal polynomials at 
$\zeta=\pm 2$.  For more details on this calculation see equation (7.72) in
\cite{Deift:1998-book} or equation (4.16) and Appendix B in \cite{Deift:1999b},
noting that the variable $\zeta$ is rescaled by a constant factor.
Then for large $\zeta$ (such as on $\partial\mathbb{D}_{a^{\star}}$) we can estimate
\eq
\begin{split}
\label{super-R-error}
{\bf R}(\zeta) & = \left({\bf I}+{\cal O}\left(\frac{1}{r\zeta}\right)\right) \left({\bf I}+\mathcal{O}\left(\frac{1}{\zeta}\right) \right) \exp\left[r \left( g_H(\zeta-\zeta_0) -\log\zeta - \sum_{j=1}^K \frac{c_j^{(H)}}{\zeta^{j}}  \right) {\bf\Lambda}_{13} \right]\\
  & = {\bf I} + \mathcal{O}\left(\frac{1}{\zeta}\right) + \mathcal{O}\left(\frac{r}{\zeta^{K+1}}\right).
\end{split}
\endeq

\subsection{Error analysis in the supercritical case}
\label{super-error}

Let $\mathbb{D}_\alpha$ and $\mathbb{D}_\beta$ be 
small, fixed, closed disks centered at $\alpha$ and $\beta$ 
that are bounded away from the outer lenses.  
Orient the boundaries $\partial\mathbb{D}_\alpha$ and 
$\partial\mathbb{D}_{\beta}$ 
clockwise.  For $z\in\mathbb{D}_\alpha$, 
let ${\bf P}_{\mbox{Ai}}^{(\alpha)}(z)$ be the {\it Airy parametrix} 
satisfying 
\begin{itemize}
\item ${\bf P}_{\mbox{Ai}}^{(\alpha)}(z)$ has the same jumps as 
${\bf W}(z)$ for $z\in\mathbb{D}_\alpha$,
\item $ \ds {\bf P}_{\mbox{Ai}}^{(\alpha)}(z){\bf \Psi}(z)^{-1} = {\bf I} + \mathcal{O}\left(\frac{1}{n}\right)$ for $z\in\mathbb{D}_\alpha$.
\end{itemize}
The construction of the Airy parametrix is standard, involving 
Airy functions and a local change of variables.  See \cite{Bleher:2004b} 
Section 7, for example, for an Airy parametrix for a $3\times3$ Riemann-Hilbert
problem.  The Airy parametrix ${\bf P}_{\mbox{Ai}}^{(\beta)}(z)$ is 
defined analogously for $z\in\mathbb{D}_\beta$.

We now define the global parametrix ${\bf \Psi}^\infty(z)$ by
\eq
{\bf \Psi}^\infty(z):=
\begin{cases}
{\bf \Psi}(z), & z\notin \mathbb{D}_\alpha \cup \mathbb{D}_\beta \cup \mathbb{D}_{a^{\star}}, \\
{\bf \Psi}(z) {\bf R}(\zeta(z)), & z\in\mathbb{D}_{a^{\star}}, \\
{\bf P}_{\mbox{Ai}}^{(\alpha)}(z), & z\in\mathbb{D}_\alpha \\
{\bf P}_{\mbox{Ai}}^{(\beta)}(z), & z\in\mathbb{D}_\beta.
\end{cases}
\endeq
The \emph{error matrix} ${\bf E}(z)$ is given by
\eq
\label{super-error-matrix}
{\bf E}(z) := {\bf W}(z){\bf \Psi}^\infty(z)^{-1}.
\endeq
Let $\Gamma$ denote the contours given by the boundaries of the regions $\Omega_j$ in Figure \ref{fig_biglensSuper}.
The error matrix satisfies a Riemann-Hilbert problem with the following 
jumps:

$\bullet$ For $z$ outside the disks $\mathbb{D}_\alpha$, $\mathbb{D}_\beta$, 
and $\mathbb{D}_{a^{\star}}$, and excluding the band $[\alpha,\beta]$:
\eq
\label{super-VE-outside}
{\bf V^{(E)}}(z) =
{\bf \Psi}(z) {\bf V^{(W)}}(z) {\bf \Psi}(z)^{-1}, \quad z\in \Gamma\cap\left( \mathbb{D}_{\alpha} \cup \mathbb{D}_{\beta} \cup \mathbb{D}_{a^{\star}} \right)^c\cap[\alpha,\beta]^c,
\endeq
where ${\bf V^{(W)}}(z)$ is given by (\ref{super-Wjumps}).

$\bullet$ For $z$ on the boundaries of the disks $\partial\mathbb{D}_\alpha$, 
$\partial\mathbb{D}_\beta$, and $\partial\mathbb{D}_{a^{\star}}$:
\eq
\label{super-VE-boundaries}
{\bf V^{(E)}}(z) =
\begin{cases}
{\bf \Psi}(z){\bf R}(\zeta) {\bf \Psi}(z)^{-1}, & z\in\partial\mathbb{D}_{a^{\star}}, \\
{\bf P}_{\mbox{Ai}}^{(\alpha)}(z) {\bf \Psi}(z)^{-1}, & z\in\partial\mathbb{D}_\alpha, \\
{\bf P}_{\mbox{Ai}}^{(\beta)}(z) {\bf \Psi}(z)^{-1}, & z\in\partial\mathbb{D}_\beta.
\end{cases}
\endeq

$\bullet$ For $z$ inside the disk $\mathbb{D}_{a^{\star}}$:
\eq
\label{super-VE-a*}
{\bf V^{(E)}}(z) = {\bf \Psi}(z) {\bf R}_-(\zeta)
\begin{pmatrix}
1 & e^{n\mathcal{P}_1} & 0 \\
0 & 1 & 0 \\
0 & 0 & 1 \end{pmatrix} {\bf R}_-(\zeta)^{-1} {\bf \Psi}(z)^{-1},
\;\; z \in \Gamma\cap\mathbb{D}_{a^{\star}}.
\endeq

$\bullet$ Furthermore, ${\bf V^{(E)}}(z)={\bf I}$ on the contours 
$$[\alpha,\beta]\cap(\mathbb{D}_\alpha\cup\mathbb{D}_\beta)^c, \quad \Gamma\cap\mathbb{D}_\alpha, \quad \text{and} \quad \Gamma\cap\mathbb{D}_\beta.$$ 
The jump contours $\Gamma^{\bf(E)}$ for ${\bf E}(z)$ are shown in Figure 
\ref{super-error-fig}.

\begin{figure}
\setlength{\unitlength}{2.7pt}
\begin{center}
\begin{picture}(100,50)(-50,-25)
\put(-10,0){\circle{6}}
\put(-12.2,-1){$\mathbb{D}_\alpha$}
\put(10,0){\circle{6}}
\put(7.9,-1){$\mathbb{D}_\beta$}
\put(28,0){\circle{6}}
\put(20,-6){$\mathbb{D}_{a^{\star}}$}
\put(23,-4){\line(2,1){5}}
\put(-50,0){\line(1,0){37}}
\put(13,0){\line(1,0){30}}
\qbezier(-9,3)(0,13),(9,3)
\qbezier(-9,-3)(0,-13),(9,-3)
\qbezier(15,0)(15,25)(-10,25)
\qbezier(15,0)(15,-25)(-10,-25)
\qbezier(-35,0)(-35,25)(-10,25)
\qbezier(-35,0)(-35,-25)(-10,-25)
\thicklines
\put(-42,0){\vector(1,0){1}}
\put(-22,0){\vector(1,0){1}}
\put(20,0){\vector(1,0){1}}
\put(28,0){\vector(1,0){1}}
\put(37,0){\vector(1,0){1}}
\put(-10,-25){\vector(1,0){1}}
\put(-10,25){\vector(-1,0){1}}
\put(0,8){\vector(1,0){1}}
\put(0,-8){\vector(1,0){1}}
\put(-10,2.8){\vector(1,0){1}}
\put(10.2,2.8){\vector(1,0){1}}
\put(28,2.8){\vector(1,0){1}}
\end{picture}
\end{center}
\caption{The jump contours $\Gamma^{\bf(E)}$ for the Riemann-Hilbert problem for ${\bf E}(z)$ in the supercritical case.\label{super-error-fig}}
\end{figure}
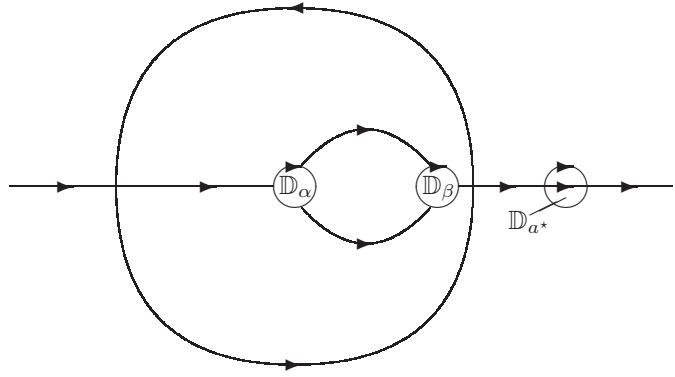

We now show that all of the jump matrices in 
\eqref{super-VE-outside}-\eqref{super-VE-a*} 
are uniformly close to the identity as $n\to\infty$.  

For the error bounds it will be convenient to split $\Gamma^{\bf(E)}$ 
into a compact component $\Gamma_C^{\bf(E)}$ and a 
noncompact component $\Gamma_N^{\bf(E)}$:
\eq
\begin{split}
\Gamma_C^{\bf(E)}:=&\partial\mathbb{D}_\alpha\cup\partial\mathbb{D}_\beta\cup\partial\mathbb{D}_{a^{\star}}\cup(\Gamma\cap\mathbb{D}_{a^{\star}}),\\
\Gamma_N^{\bf(E)}:=&\Gamma^{\bf(E)}\backslash\Gamma_C^{\bf(E)}.
\end{split}
\endeq
We now gather the results we will need on the functions $P_1(z)$, $P_2(z)$, 
and $P_3(z)$ defined by \eqref{P1-nolog}--\eqref{P3-nolog}.
\begin{lemma}
\label{P-nolog-lemma}
In the supercritical regime, the inner and outer lenses can be chosen so that:
\begin{itemize}
\item[(a)] On the inner lenses outside of the disks around $\alpha$ and 
$\beta$:  The real part of $P_1(z)$ is positive and bounded away from zero 
for 
$z\in[(\partial\Omega_2\cap\partial\Omega_3)\cup(\partial\Omega_4\cap\partial\Omega_5)]\cap(\mathbb{D}_\alpha\cup\mathbb{D}_\beta)^c$.
\item[(b)] On the real axis outside of $[\alpha,\beta]$ and the disks around 
$\alpha$ and $\beta$:  The real part of $P_1(z)$ is negative and bounded away 
from zero for 
$z\in[(\partial\Omega_1\cap\partial\Omega_6)\cup(\partial\Omega_2\cap\partial\Omega_5)]\cap(\mathbb{D}_\alpha\cup\mathbb{D}_\beta)^c$.
\item[(c)] On the real axis outside of the band $[\alpha,\beta]$ and a fixed 
distance away from $\alpha$, $\beta$, and $a^{\star}$:  The real part of $P_2(z)$ 
is negative and bounded away from zero 
for $z\in[(\partial\Omega_2\cap\partial\Omega_5)\cap(\mathbb{D}_\alpha\cup\mathbb{D}_\beta\cup\mathbb{D}_{a^{\star}})^c]\cup(\partial\Omega_1\cap\partial\Omega_6)$.
\item[(d)] On the outer lenses:  For $\kappa$ sufficiently small, the real 
part of $P_3(z)$ is negative and bounded away from zero for 
$z\in(\partial\Omega_1\cap\partial\Omega_2)\cup(\partial\Omega_5\cap\partial\Omega_6).$
\end{itemize}
\end{lemma}
\begin{proof}
Statements (a) and (b) follow from the analysis of the Riemann-Hilbert 
problem for the standard (not multiple) orthogonal polynomials 
(see, for instance, \cite{Deift:1998-book}).  Statement (c) follows from the 
definitions of the supercritical region and $a^{\star}$.  

For (d), we begin 
by choosing the outer lenses used to define ${\bf W}(z)$.  Fix 
$\kappa=0$.  Note that 
\eq
P_3(\beta) = -P_1(\beta)+P_2(\beta) = P_2(\beta) < P_2(a^{\star}) = 0.
\endeq
The second equality uses the fact that $ {\Re}\, P_1$ is zero on 
$[\alpha,\beta ]$  and $P_1$ is real for $x>\beta$; 
the inequality follows since $a^\star$ is the location of the global maximum 
of $P_2(z)$;  and the final equality is true by the choice of the Lagrange 
multiplier $l_2$.  Thus, there is a fixed radius neighborhood of $\beta$ 
in which ${\rm Re}P_3(z)<0$ for real $z$.  We choose the outer lenses to 
be a circle centered below $\alpha$ whose right-most endpoint passes through 
the real axis at some point on $(\beta,a^\star)$.  We choose the circle big enough such that $\Re P_2$ is negative on the real axis to the left of the circle.   This is always possible due to Assumption \ref{assumptionAV} (iv). 

We now show that the outer lenses are descent lines of 
${\rm Re}\,P_3(z)$.  Clearly the real part of $az$ decreases as we move to 
the left along the lenses.  Note that 
${\rm Re}\,g(z) = \int_\alpha^\beta\log|z-s|\rho_{\text{min}}(s)ds$ where 
$\rho_{\text{min}}(s)$ is the associated equilibrium measure.  Now
for any $s\in(\alpha,\beta)$, $\log|z-s|$ is increasing as $z$ moves to the 
left along the lens (one can see this clearly by drawing a circle that is 
centered at $s$ and is tangent to the outer lenses at the right-most point).  
Since $\rho_{\text{min}}(s)$ is positive for 
$s\in(\alpha,\beta)$, ${\rm Re}\,g(z)$ increases as $z$ moves to the left 
along the lenses.  This shows (d) for $\kappa=0$.  Since 
$\widetilde \alpha(\kappa)=\alpha+\mathcal{O}(\kappa)$ and 
$\widetilde \beta=\beta+\mathcal{O}(\kappa)$, (d) also holds for 
$\kappa$ sufficiently small.
\end{proof}

We now present the results we will need for the modified functions 
$\mathcal{P}_1(z;\kappa)$, $\mathcal{P}_2(z;\kappa)$, and 
$\mathcal{P}_3(z;\kappa)$ defined by 
\eqref{super-mathcalP1}--\eqref{super-mathcalP3}.  
\begin{lemma}
\label{super-Pmathcal-lemma}
For $\kappa$ sufficiently small (refer to Figure \ref{fig_biglensSuper} and Figure \ref{super-error-fig}): 
\begin{itemize}
\item[(a)] On the inner lenses outside of 
the disks around $\alpha$ and $\beta$:  The real part of 
${\mathcal P}_1(z;\kappa)$ 
is positive and bounded away from zero for
$z\in[(\partial\Omega_2\cap\partial\Omega_3)\cup(\partial\Omega_4\cap\partial\Omega_5)]\cap(\mathbb{D}_\alpha\cup\mathbb{D}_\beta)^c$.
\item[(b)] On the real axis outside of $[\alpha,\beta]$ and a fixed distance 
away from $\alpha$, $\beta$, and $a^{\star}$:  The real part of 
${\mathcal P}_1(z;\kappa)$ is negative and bounded away from zero for
$z\in[(\partial\Omega_1\cap\partial\Omega_6)\cup(\partial\Omega_2\cap\partial\Omega_5)]\cap(\mathbb{D}_\alpha\cup\mathbb{D}_\beta\cup\mathbb{D}_{a^{\star}})^c$.
\item[(c)] On the real axis outside of the outer 
lenses and a fixed distance away from $a^{\star}$:  The real part of 
$\mathcal{P}_2(z;\kappa)$ is negative and bounded 
away from zero for 
$z\in(\partial\Omega_1\cap\partial\Omega_6)\cap(\mathbb{D}_{a^{\star}})^c$.
\item[(d)] On the outer lenses:  The real part of 
$\mathcal{P}_3(z;\kappa)$ is negative and bounded away from zero for 
$z\in(\partial\Omega_1\cap\partial\Omega_2)\cup(\partial\Omega_5\cap\partial\Omega_6).$
\item[(e)] For real $z$ inside $\mathbb{D}_{a^{\star}}$:  Let $g_H(\zeta)$ be 
defined by \eqref{super-gH-def}.  Then the real part of 
$$\mathcal{P}_1(z;\kappa) +\kappa g_H(\zeta) - \kappa\log\zeta - \kappa\sum_{j=1}^K \frac{c_j^{(H)}}{\zeta^{j}}$$
is negative and bounded away from zero for 
$z\in\Gamma\cap\mathbb{D}_{a^{\star}}$.
\end{itemize}
\end{lemma}
\begin{proof}
Parts (a) through (d) follow from Lemma \ref{P-nolog-lemma}(a)--(d) 
together 
with the convergence $\mathfrak g \to g$ guaranteed in Proposition \ref{propdeform},
 the boundedness 
of $\log(z-a^{\star})$ and $(z-a^{\star})^{-j}$, $j=1,\dots,K $ outside of 
$\mathbb{D}_{a^{\star}}$, 
and $\delta_j(0)=0$.  

For part (e), first note that comparing the two expressions 
\eqref{super-mathcalP2} and \eqref{super-P2-zeta} for 
$\mathcal{P}_2(z;\kappa)$ gives 
\eq
\label{rearranging-P2}
\kappa\log(z-a^{\star}) - \kappa\log\zeta + \sum_{j=1}^{K }\frac{\delta_j}{2(z-a^{\star})^j} - \kappa \sum_{j=1}^{K } \frac{c_j^{(H)}}{\zeta^{j}} =
\frac 1 2 f(z;\kappa) -\frac 12 \mathfrak b 
-\kappa\ln \sqrt\kappa
-\frac{\kappa}{4}(\zeta-\zeta_0)^2 .
\endeq
By the choice of $\mathbb{D}_{a^{\star}}$ (see \eqref{condition-on-Da*}), for $\kappa$ 
sufficiently small we have 
\eq
\left|\Re\left[
\frac 1 2 f(z;\kappa) - \frac {\mathfrak b}2
 - \kappa\ln \sqrt\kappa
  -\frac{\kappa}{4}(\zeta-\zeta_0)^2\right]\right| < |\Re \, P_1(z)| \text{ for } z\in\mathbb{D}_{a^{\star}}.
\endeq
Write 
\eq
\begin{split}
& \hspace{-.5in} \mathcal{P}_1(z;\kappa) +\kappa g_H(\zeta) - \kappa\log\zeta - \kappa\sum_{j=1}^k\frac{c_j^{(H)}}{\zeta^{2j}} \\
& = \underbrace{-V + 2(1-\kappa)\mathfrak{g} + \ell_1 + \kappa g_H}_{=P_1 + \mathcal{O}(\kappa)} + \underbrace{\kappa\log(z-a^{\star}) - \kappa\log\zeta + \sum_{j=1}^{
K 
}\frac{\delta_j}{2(z-a^{\star})^j} - \kappa \sum_{j=1}^{
K
}\frac{c_j^{(H)}}{\zeta^{j
}}}_{\text{Has real part bounded above by }|\Re\, P_1|}.
\end{split}
\endeq
For $\kappa$ small, by
the convergence $\mathfrak g\to g$ in Proposition \ref{propdeform}, 
 the first group of terms on the 
left-hand side, 
$-V + 2(1-\kappa)\mathfrak{g} + \ell_1 + \kappa g_H$, is within $\mathcal{O}(\kappa)$ 
of $P_1$, which has strictly negative real part for $z\in\mathbb{D}_{a^{\star}}$.  Since 
the real part of the second group of terms is strictly less than the magnitude of 
the real part of $P_1$, part (e) follows.
\end{proof}
We are now in a position to bound the jumps ${\bf V^{(E)}}(z)$ of the 
error problem.
\begin{lemma}
\label{super-VE-lemma}
In the supercritical regime, for large $n$,
\begin{itemize}
\item[(a)]  Outside the disks $\mathbb{D}_\alpha$, $\mathbb{D}_\beta$, and 
$\mathbb{D}_{a^{\star}}$:  There is a constant $c>0$ such that 
$${\bf V^{(E)}}(z;\kappa) = {\bf I} + \mathcal{O}(e^{-cn}), \quad z\in\Gamma_N^{\bf(E)}.$$
\item[(b)]  On the boundary of $\mathbb{D}_{a^{\star}}$:  
$$
{\bf V^{(E)}}(z;\kappa) = {\bf I} + \mathcal{O}\left(\ds n^{-(1-\gamma)/2}\right) + \mathcal{O}\left(\ds n^{
\gamma - \frac {1-\gamma}2 (K+1)} \right), \quad  z\in\partial\mathbb{D}_{a^{\star}}.
$$
\item[(c)]  On the boundaries of $\mathbb{D}_\alpha$ and $\mathbb{D}_\beta$:  
$${\bf V^{(E)}}(z;\kappa) = {\bf I} + \mathcal{O}\left(n^{-1}\right), \quad z\in\partial\mathbb{D}_\alpha\cup\partial\mathbb{D}_\beta.$$
\item[(d)]  Inside $\mathbb{D}_{a^{\star}}$:  
$${\bf V^{(E)}}(z) = {\bf I} + \mathcal{O}(e^{-cn}), \quad z\in\Gamma\cap\mathbb{D}_{a^{\star}}.$$
\end{itemize}
\end{lemma}
\begin{proof}
Part (a) follows from \eqref{super-VE-outside}, Lemma 
\ref{super-Pmathcal-lemma}(a)--(d), and the boundedness of ${\bf \Psi}(z)$.
Part (b) follows from \eqref{super-zeta-order}, \eqref{super-R-error}, and the 
boundedness of ${\bf \Psi}(z)$.  
Part (c) comes from the construction of the parametrices 
${\bf P}_\text{\bf Ai}^{(\alpha)}(z)$ and 
${\bf P}_\text{\bf Ai}^{(\beta)}(z)$ (see, for instance, \cite{Deift:1999b}).

For part (d) we consider the jumps 
\eqref{super-VE-a*} inside the disk $\mathbb{D}_{a^{\star}}$.  
Looking at the formula 
\eqref{super-R} for ${\bf R}(\zeta)$, it appears there may be a problem at 
$\zeta=0$.  However, note that
\eq
\label{VE-bound-Gamma-a*}
\begin{split}
{\bf V^{(E)}}(z;\kappa) = & \Psi(z) e^{-\frac{r}{2}\ell_H{\bf\Lambda}_{13}} {\bf H}_{13-}(\zeta) e^{-r(g_H(\zeta)-\frac{\ell_H}{2}){\bf\Lambda}_{13}} \bpm 1 & (*)_{12} & 0 \\ 0 & 1 & 0 \\ 0 & 0 & 1 \epm \times \\
& \times e^{r(g_H(\zeta)-\frac{\ell_H}{2}){\bf\Lambda}_{13}}{\bf H}_{13-}(\zeta)^{-1}e^{\frac{r}{2}\ell_H{\bf\Lambda}_{13}}\Psi(z)^{-1}, \quad z\in\Gamma\cap\mathbb{D}_{a^{\star}}, 
\end{split}
\endeq
wherein
\eq
(*)_{12} = \exp\left(n\mathcal{P}_1(z;\kappa)+rg_H(\zeta)-r\log\zeta - r\sum_{j=1}^{K}
\frac{c_j^{(H)}}{\zeta^{j}}\right).
\endeq
Now \eqref{VE-bound-Gamma-a*} together with Lemma 
\ref{super-Pmathcal-lemma}(e) and the boundedness of ${\bf \Psi}(z)$ inside 
$\mathbb{D}_{a^{\star}}$ establishes (d).
\end{proof}

We can now show that the error matrix ${\bf E}(z)$ is uniformly close 
to the identity.
\begin{lemma}
\label{super-E}
In the supercritical regime, for $n$ large,
$${\bf E}(z) = {\bf I}+\mathcal{O}\left(n^{-(1-\gamma)/2}\right)$$
uniformly in $z$.
\end{lemma}
\begin{proof}
From Lemma \ref{super-VE-lemma}(b)--(d), 
\eq
\label{super-VE-compacta}
{\bf V^{(E)}}(z) = {\bf I} + \mathcal{O}\left(n^{-(1-\gamma)/2}\right) + \mathcal{O}\left(n^{
\gamma -\frac {1-\gamma}2 (K+1)
}\right), \quad z\in\Gamma_C^{\bf(E)}.
\endeq 
The first error term always dominates or matches the second term if the 
nonnegative integer $K$ is chosen so \eqref{super-k-def} is satisfied.
Then, for $n$ sufficiently large there 
exists a constant $c$ such that 
\eq
||{\bf V^{(E)}}-{\bf I}||_{L^2\left(\Gamma_C^{\bf(E)}\right)} + ||{\bf V^{(E)}}-{\bf I}||_{L^\infty\left(\Gamma_C^{\bf(E)}\right)} \leq c n^{-(1-\gamma)/2}.
\endeq
Also, from Lemma \ref{super-VE-lemma}(a), for $n$ sufficiently large there is 
a constant $c$ such that 
\eq
||{\bf V^{(E)}}-{\bf I}||_{L^2\left(\Gamma_N^{\bf(E)}\right)} + ||{\bf V^{(E)}}-{\bf I}||_{L^\infty\left(\Gamma_N^{\bf(E)}\right)} \leq c e^{-cn},
\endeq
The result follows by a standard technique that consists of writing the 
solution to the Riemann-Hilbert 
problem in terms of a Neumann series involving ${\bf V^{(E)}}-{\bf I}$ 
(see, for instance, \cite{Deift:1999b} Section 7.2 or 
\cite{Ercolani:2003} Section 3.5).
\end{proof}

\subsection{The supercritical kernel and proof of Theorem \ref{theorem-super-kernel}}
\label{super-kernel-proof}

\begin{proof}[Proof of Theorem \ref{theorem-super-kernel}]
Recall that the kernel is defined by (\ref{mop-kernel}):
\eq
\label{super-Kn}
K_n(x, y) = \frac{e^{-\frac{1}{2} n ( V(x) + V(y) )}}{2\pi i(x-y)} \left( \left[ {\bf Y}(y)^{-1} {\bf Y}(x) \right]_{21} + e^{n a y} \left[ {\bf Y}(y)^{-1}{\bf Y}(x)\right]_{31} \right).
\endeq
We consider local coordinates $\zeta_x$ and $\zeta_y$ in $\mathbb{D}_{a^{\star}}$.
While the function ${\bf Y}(z)$ has a 
jump in this region, the first column of ${\bf Y}(z)$ does not (see the 
Riemann-Hilbert problem \eqref{rhp}).  Therefore we can pick $x$ and $y$ to 
be in a convenient region.  We choose $x$ and $y$ to be in $\Omega_1$ as 
defined in Figure \ref{fig_biglensSuper}.  

From the transformation \eqref{super-Y-to-W}, we see that 
\bea
\label{super-Yinv-Y21}
&& \begin{split}
& \left[{\bf Y}(y)^{-1}{\bf Y}(x)\right]_{21} = \left[{\bf W}(y)^{-1}{\bf W}(x)\right]_{21}  \exp\left(n\left((1-\kappa){\mathfrak g}(y)+(1-\kappa){\mathfrak g}(x)+
{\delta V(x)}
 +\ell_1\right)\right), 
\end{split}
\\
\label{super-Yinv-Y31}
&& \begin{split}
\left[{\bf Y}(y)^{-1}{\bf Y}(x)\right]_{31} = \left[{\bf W}(y)^{-1}{\bf W}(x)\right]_{31}  \exp\left(n\left((1-\kappa){\mathfrak g}(x)+
{\delta V(x)}  +
{\delta V(y)} +
l_2 - \eta
\right)\right)
\end{split}
\eea
for $x$ and $y$ in $\Omega_1$.  We have 
\eq
\label{super-W-Psiinf}
{\bf W}(z) =  {(\1 +\mathcal{O}(n^{-(1-\gamma)/2})  }{\bf \Psi}^\infty(z)  = 
  {(\1 +\mathcal{O}(n^{-(1-\gamma)/2})  } 
{\bf \Psi}(z){\bf R}(\zeta(z)) 
\endeq
for $z\in\mathbb{D}_{a^{\star}}$.  From \eqref{super-zeta}, we have 
\eq
\label{Psiinv-Psi}
{\bf \Psi}(y)^{-1} {\bf \Psi}(x) = {\bf I} + \mathcal{O}\left((\zeta_x - \zeta_y)\kappa^{1/2}\right).
\endeq
We define the functions 
$\mathcal{Q}_i(z;\kappa)$ to be:
\begin{eqnarray}
{\mathcal Q}_1(z;\kappa) & := & -V(z) + 2(1-\kappa){\mathfrak g}(z;\kappa) + \ell_1(\kappa), \\
{\mathcal Q}_2(z;\kappa) & := & -V(z) + az + (1-\kappa){\mathfrak g}(z;\kappa) + 
l_2, \\
{\mathcal Q}_3(z;\kappa) & := & az - (1-\kappa){\mathfrak g}(z;\kappa) -\ell_1(\kappa) + 
{l_2}.
\end{eqnarray}
Now combining \eqref{super-Yinv-Y21}, 
\eqref{super-W-Psiinf}, \eqref{Psiinv-Psi}, \eqref{super-R}, 
$\det{\bf R}(\zeta)=1$, and 
noting the $\mathcal{O}(\kappa)$ error terms from \eqref{Psiinv-Psi} 
are subsumed by the $\mathcal{O}(n^{-(1-\gamma)/2})$ error terms from 
\eqref{super-W-Psiinf} gives
\eq
\label{super-YinvY21}
e^{-\frac{n}{2}(V(x)+V(y))}\left[{\bf Y}(y)^{-1}{\bf Y}(x)\right]_{21} = \left(\mathcal{O}\left((\zeta_x-\zeta_y)\kappa^{1/2}\right)\cdot H_r^{(r)}(\zeta_x-\zeta_0) + \mathcal{O}(n^{-(1-\gamma)/2}) \right) e^{n(*)},
\endeq
where
\eq
\begin{split}
(*) = & \frac{1}{2}\mathcal{Q}_1(x) + \frac{1}{2}\mathcal{Q}_1(y) + 
{\delta V(x)}
-\kappa\log\zeta_x - \kappa\sum_{j=1}^{K} \frac{c_j^{(H)}}{\zeta_x^j} \\
    = & \frac{1}{2}P_1(x) + \frac{1}{2}P_1(y) + \mathcal{O}(\kappa\ln \kappa).
\end{split}
\endeq
The last equality is shown by noticing that rearranging the terms in \eqref{mainidzeta}  we have 
\be
{\delta V(x)} 
 -\kappa\log\zeta_x - \kappa\sum_{j=1}^{K} \frac{c_j^{(H)}}{\zeta_x^j} = 
\frac 1 2\le(V(x) - a x - (1-\kappa)\mathfrak g(x)  - \ell_2 - \frac \kappa 2(\zeta-\zeta_0)^2  + \kappa \ln \kappa + \mathfrak b \ri)
\label{125}
\ee

Since the real  part of $P_1(z)$ is negative for $z$ near $a^{\star}$ (Lemma 
\ref{P-nolog-lemma}(b)), for $\kappa$ sufficiently 
small the real part of the exponent in \eqref{super-YinvY21} is negative.  

Define
\eq
F_r^\text{GUE}(\zeta_x,\zeta_y):=\frac{2\pi i}{k_{r-1}^{(r)}}\left(H_r^{(r)}(\zeta_x-\zeta_0)H_{r-1}^{(r)}(\zeta_y-\zeta_0)-H_{r-1}^{(r)}(\zeta_x-\zeta_0)H_r^{(r)}(\zeta_y-\zeta_0)\right).
\endeq
From \eqref{super-Yinv-Y31},
\eq
e^{-\frac{n}{2}(V(x)+V(y))+nay}\left[{\bf Y}(y)^{-1}{\bf Y}(x)\right]_{31} = \left(F_r^\text{GUE}(\zeta_x,\zeta_y)+\mathcal{O}(n^{-(1-\gamma)/2})  \right) e^{n(**)},
\endeq
where from \eqref{super-Yinv-Y31}
\bea 
\begin{split}
(**) &= -\frac {V(x)}2 - \frac {V(y)}2  + ay + \delta V(y) + \delta  V(x)
 - \kappa \ln \zeta_x  - \kappa \sum_{j=1}^K \frac {c_j^{(H)}}{\zeta_x^j} 
 - \kappa \ln \zeta_y  - \kappa \sum_{j=1}^K \frac {c_j^{(H)}}{\zeta_y^j}  +
\\ &\phantom{=} + (1-\kappa) \mathfrak g(x) + \kappa \ell_H\,.
\end{split}
\eea
Using now \eqref{125} (for both $x$ and $y$) and rearranging the terms we find 
\bea
(**) =&& \frac 1 2 \le[ ay - (1-\kappa) \mathfrak g(y)  - \frac \kappa 2 (\zeta_y-\zeta_0)^2 - l_2 + \kappa \ln \kappa  + \mathfrak b  \ri] 
+ \nonumber\\
&&+
\frac 1 2 \le[  (1-\kappa) \mathfrak g(x) -ax - \frac \kappa 2 (\zeta_x-\zeta_0)^2 - l_2 + \kappa \ln \kappa  + \mathfrak b  \ri] 
+ \kappa \ell_H - \eta \\
=&&  -\frac{\kappa}{4}(\zeta_x-\zeta_0)^2-\frac{\kappa}{4}(\zeta_y-\zeta_0)^2-\frac{1}{2}\mathcal{Q}_3(x)+\frac{1}{2}\mathcal{Q}_3(y).
\eea
where we have used the definition of $\eta := \kappa \ln \kappa + \mathfrak b + \kappa \ell_H$ \eqref{defeta}.
 From this and the 
fact that \eqref{super-YinvY21} is exponentially decaying in $n$ shows
\eq
\label{super-kernel-formula}
K_n(x,y) = \frac{e^{-\frac{n}{2}\mathcal{Q}_3(x)+\frac{n}{2}\mathcal{Q}_3(y)}}{2\pi i (x-y)}\left(F_r^\text{GUE}(\zeta_x,\zeta_y)+\mathcal{O}(n^{-(1-\gamma)/2})\right)e^{-\frac{r}{4}(\zeta_x-\zeta_0)^2-\frac{r}{4}(\zeta_y-\zeta_0)^2}.
\endeq
One final application of \eqref{super-zeta} (to switch $x-y$ to 
$\zeta_x-\zeta_y$) and 
Proposition  \ref{propdeform} (to show convergence of $\mathcal{Q}_3$ to $P_3$) gives 
\eqref{kernel-thm-form}.

\end{proof}

\section{The subcritical regime}

We now take $V(x)$ and $a$ so Definition
\ref{subcritical} of the subcritical regime is satisfied; that is $a<a_c$ and the function $\Re\,P_2$ has no global maximum on $\mathbb R \setminus [\alpha,\beta]$. In this case $\Re\, P_3 $ has a (unique) global minimum at $z=b^\star$  (with value zero, as per our choice of $l_2$ in Definition \ref{defell2}).
  We will show that almost surely there are no outliers.
We will freely reuse the same notation from Section \ref{supercrit} 
for new objects which played a similar role in the analysis of the 
supercritical regime.  To begin, fix $\gamma\in[0,1)$ and again set $K$ 
to be the smallest nonnegative integer satisfying
\eq
\label{sub-k-def}
K\geq \max \le\{\frac{3\gamma-1}{1-\gamma},0\ri\}
\endeq
\subsection{Modified equilibrium problem (subcritical case)}
The procedure here parallels closely the one followed in the supercritical case, and hence we will only state the results since their proof does not differ significantly from the other case.

Let $J$ be a closed subset not containing the point $b^\star$ and containing $[\alpha,\beta]$ in its interior; we recall that $b^\star(a)>\beta$ for $0<a<a_c$.

\bp
\label{propdeform-sub}
For any $K\in \mathbb N$ there is a neighborhood of the origin in $(\kappa,\vec \delta) \in \C^{1+K}$  such that the equilibrium measure $\widetilde \sigma(x)\d x$ of {\bf unit} total mass for the external field 
\be
\widetilde V(z):=V(z)  +  \delta V(z), \ \ \ \delta V(z):= \kappa \ln (z-b^\star ) + \kappa \sum_{j=1}^K \frac {\delta_j}{2(z-b^\star)^j}
\label{deformV-sub}
\ee
is supported on a single interval $[\alpha(\kappa,\vec \delta), \beta (\kappa,\vec \delta)]$ still contained in the interior of $J$: the endpoints $\alpha(\kappa,\vec \delta), \beta (\kappa,\vec \delta)$ are analytic functions of the specified variables.
Furthermore the $g$--function of this problem 
\be
\mathfrak g(z):= \int \ln(z-w) \widetilde  \sigma(w) \d w
\ee
converges uniformly over closed subsets not containing $[\alpha,\beta]$ to the unperturbed $g$--function.
\ep 
The proof is identical to that of Proposition \ref{propdeform}: the only difference is that now the modified equilibrium measure is of unit total mass, rather than of mass $1-\kappa$.
We next re-define the three functions $\mathcal P_j$'s; the definition is subtly different from the previous \eqref{super-mathcalP1}, \eqref{super-mathcalP2},\eqref{super-mathcalP3} and hence there is a possibility of confusion for the reader. The advantage is that we will be able to recycle many of the previous computations.
\begin{shaded}
\begin{eqnarray}
\label{sub-mathcalP1}
{\mathcal P}_1(z) & := & -\widetilde V(z) + 2{\mathfrak g}(z) + \ell_1, \\
\label{sub-mathcalP2}
{\mathcal P}_2(z) & := & {\mathcal P}_1(z) + {\mathcal P}_3(z) = -V(z)+\delta V(z) + a z + \mathfrak g(z) +l_2 \\
\label{sub-mathcalP3}
{\mathcal P}_3(z) & := & az - {\mathfrak g}(z) + 2\delta V(z) + l_2 - \ell_1 = \mathcal P_2(z)- \mathcal P_1(z)\\
\widetilde V(z)&:=&  V(z) + \delta V(z)\ ,\qquad 
\delta V(z) :=  \kappa\log(z-b^{\star}) + \sum_{j=1}^{2k}\frac{\delta_j}{2(z-b^{\star})^j} + \ell_1
\end{eqnarray}
\end{shaded}

 It is {\bf important} to point out  the change of sign in the definition of $\widetilde V$, relative to the supercritical case. We also remind that $\ell_1, \mathfrak g$ are analytic functions of $\kappa, \vec \delta$, while $l_2$ is the constant mandated in Definition \ref{defell2}.
\begin{thm}
\label{rhoprop-sub}
There exists a conformal change of coordinate $\rho = \rho(z;\kappa,\vec \delta)$ {\em fixing $z=b^\star$} ($\rho(b^\star;\kappa,\vec \delta) \equiv 0$)  that depends analytically on the parameters $\kappa, \vec \delta$ such that 
\be
\mathcal P_3 (z):= az  - \mathfrak g(z;\kappa,\vec \delta) +l_2-\ell_1 +2\kappa\ln(z-b^\star) + \sum_{j=1}^K \frac {\delta_j}{(z-b^\star)^j} 
\label{38}
\ee
can be written as 
\be
\mathcal P_3(z;\kappa,\vec \delta) = \frac 1 2 (\rho- {\mathfrak a} )^2 + 2\kappa \ln \rho+\mathfrak b  +  \sum_{j=2}^{K} \frac 
{ \gamma_j}{\rho^j}\label{mainid-sub}
\ee
where the parameters ${\mathfrak a} = {\mathfrak a} (\kappa,\vec \delta)$, $\mathfrak b  = \mathfrak b (\kappa,\vec\delta)$ and $\vec \gamma = \vec \gamma(\kappa;\vec \delta)$ are analytic functions of the indicated parameters.
Furthermore the Jacobian 
\be
\frac {\pa \vec \gamma}{\pa \vec \delta} 
\ee 
is nonsingular in a neighborhood of the origin (for $\kappa$ sufficiently small).
\end{thm}
\begin{thm}
\label{confchange-sub}
There exists a conformal change of coordinate $\zeta(z;\kappa)$ of the form 
\be
\zeta(z;\kappa)= \frac {\rho(z;\kappa)}{i \sqrt{\kappa}} = \frac 1{i \sqrt{\kappa}} C (z-b^\star) (1+\mathcal O(z-b^\star))\ , C>0
\label{sub-zeta}
\ee
and a choice of $\vec \delta = \vec \delta(\kappa)$ for the deformed potential (\ref{deformV}) in Puiseux series of $\sqrt{\kappa}$ 
such that 
\be
\mathcal P_3(z; \kappa, \vec \delta(\kappa)) =   - \frac \kappa 2 (\zeta - \zeta_0)^2 + 2 \kappa \ln (\sqrt \kappa \zeta) +\mathfrak b  + 2 \kappa \sum_{j=1}^K \frac {c^{(H)}_j}{\zeta^j}\label{312}\,.
\ee
The functions  $\zeta_0(\kappa), \beta(\kappa), \vec \delta(\kappa)$ admit a Puiseux expansion and are of orders 
\be
\zeta_0 = \mathcal O(\sqrt\kappa)\ ,\ \ \beta = \mathcal O(\kappa)\ ,\ \ \vec \delta = \mathcal O(\kappa).
\ee
The expressions $c^{(H)}_j$ are polynomials of degree $j$ in $\zeta_0$ determined by the formula \eqref{super-cH-def}.
\end{thm}
\begin{figure}
\setlength{\unitlength}{2.7pt}
\begin{center}
\begin{picture}(100,50)(-60,-25)
\put(-10,0){\circle*{1}}
\put(-11.5,2){$\alpha$}
\put(10,0){\circle*{1}}
\put(9.5,2){$\beta$}
\put(15,0){\circle*{1}}
\put(16,1){$b^{\star}$}
\put(-50,0){\line(1,0){80}}
\put(-42,14){$\Omega_1$}
\put(-22,8.5){$\Omega_2$}
\put(-1.5,3){$\Omega_3$}
\put(-1.5,-4.5){$\Omega_4$}
\put(-22,-9.5){$\Omega_5$}
\put(-42,-15){$\Omega_6$}
\qbezier(-10,0)(0,16),(10,0)
\qbezier(-10,0)(0,-16),(10,0)
\qbezier(15,0)(15,25)(-10,25)
\qbezier(15,0)(15,-25)(-10,-25)
\qbezier(-35,0)(-35,25)(-10,25)
\qbezier(-35,0)(-35,-25)(-10,-25)
\thicklines
\put(-42,0){\vector(1,0){1}}
\put(-22,0){\vector(1,0){1}}
\put(0,0){\vector(1,0){1}}
\put(22,0){\vector(1,0){1}}
\put(-10,-25){\vector(1,0){1}}
\put(-10,25){\vector(-1,0){1}}
\put(0,8){\vector(1,0){1}}
\put(0,-8){\vector(1,0){1}}
\end{picture}
\end{center}
\caption{The regions $\Omega_i$ and the oriented contour $\Gamma$ for the subcritical case.\label{sublens}}
\end{figure}
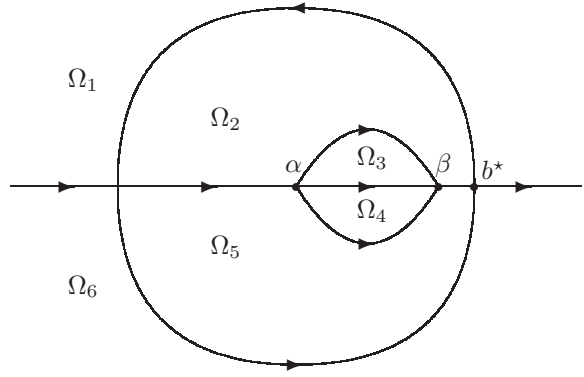

\subsection{Steepest descent analysis (subcritical case)}
\label{sub-steepest}
The regions $\Omega_i$, $i=1,...6$ are defined in Figure \ref{sublens}.
Following the example of the supercritical case, we  introduce the {\bf same} new matrix $\mathbf W$ as in \eqref{super-Y-to-W}. Note, however that the definition of the regions $\Omega_j$'s now follows Figure \ref{sublens}, the constant $l_2$ follows Definition \ref{defell2} in the subcritical case and $\delta V$ is given in \eqref{deformV-sub} instead.


The new matrix  ${\bf W}(z;\kappa)$ satisfies the same jump conditions \eqref{super-Wjumps} 
as in  the supercritical case (with, of course, the new definitions of 
the $\mathcal{P}_j(z;\kappa)$'s and $\Omega_j$'s) as well as the same 
asymptotic condition \eqref{super-W-infinity}.  
Due to the new definitions of $\mathcal P_j$'s \eqref{sub-mathcalP1}, \eqref{sub-mathcalP2}, \eqref{sub-mathcalP3} , the behavior near $z=b^{\star}$ is now different:
\eq
\label{newWasympt}
\begin{split}
{\bf W} = \text{(analytic)}\bpm 1 & 0 & 0 \\ 0 &
{ {\rm e}^{-n \delta V(z)}}
& 0 \\ 0 & 0 & 
{ {\rm e}^{n \delta V(z)}}
 \epm \\ 
\text{as } z\to b^{\star}.
\end{split}
\endeq
The outer parametrix problem is the same as in the supercritical case, so we
again define the outer parametrix solution ${\bf \Psi}(z;\kappa)$ as in 
\eqref{Vpsi}.
\subsubsection{The local parametrix near \texorpdfstring{$\bs{b^{\star}}$}{bstar}}

 Define $\mathbb{D}_{b^{\star}}$ to be 
a fixed-size circular disk centered at $b^{\star}$ which is small enough 
so that it  does not intersect the 
inner lenses, and 
\eq
\label{condition-on-Db*}
\Re(P_2)<0 \quad \text{for} \quad z\in\mathbb{D}_{b^{\star}}.
\endeq
This last condition is possible for $\kappa$ sufficiently small because --due to the definition of $l_2$ in Definition \ref{defell2} for the subcritical case--
 $P_2(b^{\star})<0$.
We now use Theorem \ref{confchange-sub}; in a local coordinate centered at $b^{\star}$, the analytic part of 
$\mathcal{P}_3(z;\kappa)$ behaves the same way (quadratically with a 
maximum at the origin) along the imaginary axis 
as $\mathcal{P}_2(z;\kappa)$ did along the real axis in the supercritical 
regime.  
The scaling of $\zeta$ is analogous to \eqref{super-zeta-order};
\eq
\label{sub-zeta-order}
\zeta=\mathcal{O}\left(n^{(1-\gamma)/2}\right) \quad \text{when } z\in\partial\mathbb{D}_{b^{\star}}.
\endeq
\eq
{\mathcal P}_3 = f(z; \kappa) + 2 \delta V(z)  +l_2-\ell_1=  -\frac{\kappa}{2}(\zeta-\zeta_0)^2 + 2\kappa\log\zeta + 2\kappa\sum_{j=1}^k\frac{c_j^{(H)}}{\zeta^{2j}}  + \kappa \ln \kappa + \mathfrak b.
\endeq

\begin{defn}
\label{localparam-sub}
The local parametrix within the disk $\mathbb D_{b^\star}$ shall be the unique solution ${\bf R}(z)$ to the following model Riemann--Hilbert problem:
\eq
\label{sub-R-RHP}
\begin{cases}
& {\bf R}_+(\zeta) = {\bf R}_-(\zeta)\bpm 1 & 0 & 0 \\ 0 & 1 & -\zeta^{2r}\exp\left(-\frac{r}{2}(\zeta-\zeta_0)^2+r\ell_H +2r\sum_{j=1}^k\frac{c_j^{(H)}}{\zeta^{2j}}\right) \\ 0 & 0 & 1 \epm  
\\ & \phantom{{\bf R}_+(\zeta)} = {\bf R}_-(\zeta)\bpm 1 & 0 & 0 \\ 0 & 1 & -e^{n\mathcal{P}_3} \\ 0 & 0 & 1 \epm, \quad \zeta\in\mathbb{R},
\\ &{\bf R}(\zeta)={\bf I}+{\cal O}\left(\ds\frac{1}{\zeta}\right) \text{ as } \zeta\to\infty,
\\&{\bf R}(\zeta)=(\mbox{analytic})\bpm 1 & 0 & 0 \\ 0 & \zeta^{-r}\exp\left(-r\sum_{j=1}^k\frac{c_j^{(H)}}{\zeta^{2j}}\right) & 0 \\ 0 & 0 & \zeta^r\exp\left(r\sum_{j=1}^k\frac{c_j^{(H)}}{\zeta^{2j}}\right) \epm \text{ as } \zeta\to 0.
\end{cases}
\endeq
\end{defn}
This problem is almost the same as \eqref{R-RHP}.  Analogously to 
\eqref{super-R}, the solution is ($\ell_H = -1 -2\ln 2$ as in \eqref{super-ellH})
\eq
\label{sub-R}
{\bf R}(\zeta) =  \exp\left(-\frac{r}{2}\ell_H{ \bf\Lambda}_{23}\right) \,{\bf H}_{23}(\zeta) \,\zeta^{-r{\bf\Lambda}_{23}}\,\exp\left(\left(\frac{r}{2}
\ell_H - r\sum_{j=1}^k\frac{c_j^{(H)}}{\zeta^{2j}}\right){\bf\Lambda}_{23}\right), 
\endeq
where
\eq
\label{sub-H23}
{\bf H}_{23}(\zeta):=\bpm 1 & 0 & 0 \\ 0 & H^{(r)}_r(\zeta-\zeta_0) & \displaystyle \frac{-1}{2\pi i}\int_{-\infty}^\infty \frac{H^{(r)}_r(s-\zeta_0) e^{-\frac{r}{2}s^2}}{s-\zeta}ds \\ \\ 0 & \displaystyle\frac{2\pi i}{-k^{(r)}_{r-1}}H^{(r)}_{r-1}(\zeta-\zeta_0)& \displaystyle \frac{1}{k^{(r)}_{r-1}}\int_{-\infty}^\infty \frac{H^{(r)}_{r-1}(s-\zeta_0) e^{-\frac{r}{2}s^2}}{s-\zeta}ds \epm \quad \text{and} \quad 
{\bf\Lambda}_{23}:=\bpm 0 & 0 & 0 \\ 0 & 1 & 0 \\ 0 & 0 & -1 \epm.
\endeq
Again the polynomials $H_m^{(r)}(\zeta)$ and the normalization constants 
$k_m^{(r)}$ are defined by \eqref{def-of-kr}.  The analysis in the 
supercritical regime leading to \eqref{super-R-error} 
applies here as well, leading to 
\eq
\label{sub-R-error}
{\bf R}(\zeta) = {\bf I} + \mathcal{O}\left(\frac{1}{\zeta}\right) + \mathcal{O}\left(\frac{r}{\zeta^{2k+2}}\right).
\endeq

\subsection{The subcritical error analysis}
Let $\mathbb{D}_\alpha$ and $\mathbb{D}_\beta$ be small, closed disks of 
fixed radii centered at $\alpha$ and $\beta$ that are 
bounded 
away from the outer lenses and $\mathbb{D}_{b^{\star}}$.  Orient the boundaries 
$\partial\mathbb{D}_\alpha$ and $\partial\mathbb{D}_\beta$ clockwise.  
Let ${\bf P}_{\mbox{Ai}}^{(\alpha)}$ and ${\bf P}_{\mbox{Ai}}^{(\beta)}$ 
be the Airy parametrices constructed in $\mathbb{D}_\alpha$ 
and $\mathbb{D}_\beta$, respectively (see Section \ref{super-error}).
Define the global parametrix ${\bf \Psi}^\infty(z)$ by
\eq
{\bf \Psi}^\infty(z):=
\begin{cases}
{\bf \Psi}(z), & z\notin \mathbb{D}_\alpha \cup \mathbb{D}_\beta \cup \mathbb{D}_{b^{\star}}, \\
{\bf \Psi}(z) {\bf R}(\zeta(z)), & z\in\mathbb{D}_{b^{\star}}, \\
{\bf P}_{\mbox{Ai}}^{(\alpha)}(z), & z\in\mathbb{D}_\alpha, \\
{\bf P}_{\mbox{Ai}}^{(\beta)}(z), & z\in\mathbb{D}_\beta.
\end{cases}
\endeq
The error matrix ${\bf E}(z)$ is given by
\eq
{\bf E}(z) := {\bf W}(z){\bf \Psi}^\infty(z)^{-1}.
\endeq
Let $\Gamma$ denote the contours given by the boundaries of the regions 
$\Omega_j$ in Figure \ref{sublens}.
The error matrix satisfies a Riemann-Hilbert problem with jump matrix 
${\bf V^{(E)}}(z)$ on the contours $\Gamma^{\bf(E)}$ shown in 
Figure \ref{sub-error-fig}.  The form of the jump matrix is as follows: \\
\indent
$\bullet$ For $z$ outside the disks $\mathbb{D}_\alpha$, $\mathbb{D}_\beta$, 
and $\mathbb{D}_{b^{\star}}$, and excluding the band $[\alpha,\beta]$:
\eq
\label{sub-VE-outside}
{\bf V^{(E)}}(z) =
{\bf \Psi}(z) {\bf V^{(W)}}(z) {\bf \Psi}(z)^{-1}, \quad z\in \Gamma\cap\left( \mathbb{D}_{\alpha} \cup \mathbb{D}_{\beta} \cup \mathbb{D}_{b^{\star}} \right)^c\cap[\alpha,\beta]^c,
\endeq
where ${\bf V^{(W)}}(z)$ is given by the formulas in (\ref{super-Wjumps}).\\
\indent
$\bullet$ For $z$ on the boundaries of the disks 
$\partial\mathbb{D}_\alpha$, 
$\partial\mathbb{D}_\beta$, and $\partial\mathbb{D}_{b^{\star}}$:
\eq
\label{sub-VE-boundaries}
{\bf V^{(E)}}(z) =
\begin{cases}
{\bf \Psi}(z){\bf R}(\zeta) {\bf \Psi}(z)^{-1}, & z\in\partial\mathbb{D}_{b^{\star}}, \\
{\bf P}_{\mbox{Ai}}^{(\alpha)}(z) {\bf \Psi}(z)^{-1}, & z\in\partial\mathbb{D}_\alpha, \\
{\bf P}_{\mbox{Ai}}^{(\beta)}(z) {\bf \Psi}(z)^{-1}, & z\in\partial\mathbb{D}_\beta.
\end{cases}
\endeq
\indent
$\bullet$ For $z$ inside the disk $\mathbb{D}_{b^{\star}}$:
\eq
\label{sub-VE-b*1}
{\bf V^{(E)}}(z) =
\begin{cases}
{\bf \Psi}(z) {\bf R}(\zeta)
\begin{pmatrix}
1 & e^{n\mathcal{P}_1(z)} & e^{n \mathcal{P}_2(z)} \\
0 & 1 & 0 \\
0 & 0 & 1 \end{pmatrix} {\bf R}(\zeta)^{-1} {\bf \Psi}(z)^{-1},
 & z \in (\partial\Omega_1\cap\partial\Omega_6)\cap\mathbb{D}_{b^{\star}}, \\
{\bf \Psi}(z) {\bf R}(\zeta)
\begin{pmatrix}
1 & e^{n\mathcal{P}_1(z)} & 0 \\
0 & 1 & 0 \\
0 & 0 & 1 \end{pmatrix} {\bf R}(\zeta)^{-1} {\bf \Psi}(z)^{-1},
 & z \in (\partial\Omega_2\cap\partial\Omega_5)\cap\mathbb{D}_{b^{\star}}.
\end{cases}
\endeq
\indent
$\bullet$ Furthermore, ${\bf V^{(E)}}(z)={\bf I}$ 
on the contours
$$[\alpha,\beta]\cap(\mathbb{D}_\alpha\cup\mathbb{D}_\beta)^c, \quad \Gamma\cap\mathbb{D}_\alpha, \quad \Gamma\cap\mathbb{D}_\beta, \quad (\partial\Omega_1\cap\partial\Omega_2)\cap\mathbb{D}_{b^{\star}}, \quad \text{and} \quad (\partial\Omega_5\cap\partial\Omega_6)\cap\mathbb{D}_{b^{\star}}.$$
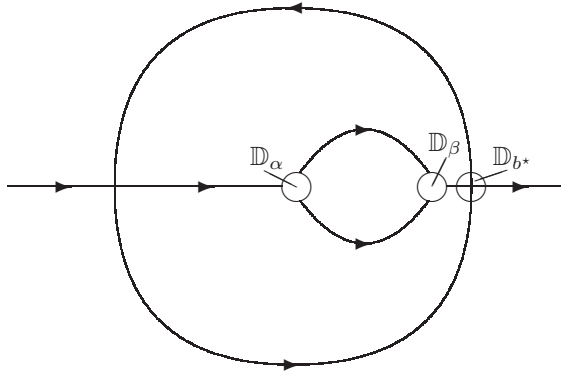
\begin{figure}
\setlength{\unitlength}{2.7pt}
\begin{center}
\begin{picture}(100,50)(-60,-25)
\put(-9.5,0){\circle{4}}
\put(-16,3){$\mathbb{D}_\alpha$}
\put(-14,2.25){\line(2,-1){4.5}}
\put(9.5,0){\circle{4}}
\put(9,5){$\mathbb{D}_\beta$}
\put(11.5,4){\line(-1,-2){2}}
\put(15,0){\circle{4}}
\put(18,3){$\mathbb{D}_{b^{\star}}$}
\put(20,2.25){\line(-3,-1){4.5}}
\put(-50,0){\line(1,0){38.5}}
\put(11.5,0){\line(1,0){16}}
\qbezier(-9,2)(0,14),(9,2)
\qbezier(-9,-2)(0,-14),(9,-2)
\qbezier(15,0)(15,25)(-10,25)
\qbezier(15,0)(15,-25)(-10,-25)
\qbezier(-35,0)(-35,25)(-10,25)
\qbezier(-35,0)(-35,-25)(-10,-25)
\thicklines
\put(-42,0){\vector(1,0){1}}
\put(-22,0){\vector(1,0){1}}
\put(22,0){\vector(1,0){1}}
\put(-10,-25){\vector(1,0){1}}
\put(-10,25){\vector(-1,0){1}}
\put(0,8){\vector(1,0){1}}
\put(0,-8){\vector(1,0){1}}
\end{picture}
\end{center}
\caption{The jump contours $\Gamma^{\bf(E)}$ for the Riemann-Hilbert problem for ${\bf E}(z)$ in the subcritical case.\label{sub-error-fig}}
\end{figure}
We now show that all of the jump matrices in 
\eqref{sub-VE-outside}--\eqref{sub-VE-b*1} are 
uniformly close to the identity as $n\to\infty$.  

%
%
%
%
%
%
Here are the results we will need for $P_1$, $P_2$, and $P_3$, analogous 
to Lemma \ref{P-nolog-lemma}:
\begin{lemma}
\label{sub-nolog-lemma}
In the subcritical regime, the inner and outer lenses can be chosen such 
that 
\begin{itemize}
\item[(a)] On the inner lenses outside of the disks around $\alpha$ and 
$\beta$:  The real part of $P_1(z)$ is positive and bounded away from zero 
 for 
$z\in[(\partial\Omega_2\cap\partial\Omega_3)\cup(\partial\Omega_4\cap\partial\Omega_5)]\cap(\mathbb{D}_\alpha\cup\mathbb{D}_\beta)^c$.
\item[(b)]  On the real axis outside of $[\alpha,\beta]$ and the disks 
around 
$\alpha$ and $\beta$:  The real part of $P_1(z)$ is negative and bounded away from zero for $z\in[(\partial\Omega_1\cap\partial\Omega_6)\cup(\partial\Omega_2\cap\partial\Omega_5)]\cap(\mathbb{D}_\alpha\cup\mathbb{D}_\beta)^c$.
\item[(c)]  On the outer lenses outside of the disk around $b^{\star}$:  For 
$\kappa$ sufficiently small, the real part of $P_3(z)$ is negative and 
bounded away from zero for 
$z\in[(\partial\Omega_1\cap\partial\Omega_2)\cup(\partial\Omega_5\cap\partial\Omega_6)]\cap\mathbb{D}_{b^{\star}}^c$.
\item[(d)]  On the real axis outside of the outer lenses or on the 
real axis inside 
$\mathbb{D}_{b^{\star}}$:  For $\kappa$ 
sufficiently small, the real part of $P_2(z)$ is negative and bounded away 
from zero for $z\in(\partial\Omega_1\cap\partial\Omega_6) \cup [(\partial\Omega_2\cap\partial\Omega_5)\cap\mathbb{D}_{b^{\star}}]$.
\end{itemize}
\end{lemma}
\begin{proof}
Parts (a) and (b) follow from the analysis of the Riemann-Hilbert problem 
for the standard (non-multiple) orthogonal polynomials 
(for example, \cite{Deift:1998-book}).  
\\
\indent
For (c), first note $P_3(b^{\star})=0$.  Following the proof of Lemma 
\ref{P-nolog-lemma}(d), 
the outer lenses (defined in this regime to be a 
circle centered below $\alpha$ and passing through $b^{\star}$, that is big enough such that $\Re P_2$ is negative on the real axis to the left of the circle) are descent lines of 
${\rm Re}\,P_3(z)$ for $\kappa$ sufficiently small.  The result follows. 
\\
\indent
For (d), start with $\kappa=0$.  Consider real $z$ to the left of the outer 
lenses.  From Lemma \ref{sub-nolog-lemma}(c), ${\rm Re}\,P_3<0$ 
at the left-most point of the outer lenses.  Thus ${\rm Re}\,P_3(z)<0$ 
for such $z$ 
since ${\rm Re}\,P_3$ is a strictly increasing function for 
$z\in(-\infty,\alpha)$.  Since ${\rm Re}\,P_1(z)$ is also negative here 
by construction, this means  
${\rm Re}\,P_2(z)={\rm Re}(P_1(z)+P_3(z))$ is also 
negative.  This is also true for real $z$ inside $\mathbb{D}_{b^{\star}}$ by 
\eqref{condition-on-Db*}.  
Next, consider real $z$ to the right of the outer lenses.  
By Definition \ref{subcritical}
 we have 
${\rm Re}\,P_2(z)<{\rm Re}\,P_3(b^{\star})=0$.   
Along with the fact that ${\rm Re}\,P_1(z)<0$ here shows the desired 
result.  
\end{proof}
Next come the necessary results for $\mathcal{P}_1(z;\kappa)$, 
$\mathcal{P}_2(z;\kappa)$, and $\mathcal{P}_3(z;\kappa)$ defined by 
\eqref{sub-mathcalP1}--\eqref{sub-mathcalP3}.  This lemma is analogous 
to Lemma \ref{super-Pmathcal-lemma} for the supercritical regime.
\begin{lemma}
\label{sub-Pmathcal-lemma}
For $\kappa$ sufficiently small:  
\begin{itemize}
\item[(a)] On the two inner lenses outside of 
the disks around $\alpha$ and $\beta$:  The real part of 
${\mathcal P}_1(z;\kappa)$ is positive and bounded away from zero for
$z\in[(\partial\Omega_2\cap\partial\Omega_3)\cup(\partial\Omega_4\cap\partial\Omega_5)]\cap(\mathbb{D}_\alpha\cup\mathbb{D}_\beta)^c$.
\item[(b)]  On the real axis outside of 
$[\alpha,\beta]$ and the disks around 
$\alpha$, $\beta$, and $b^{\star}$:  The real part of 
${\mathcal P}_1(z;\kappa)$ is negative and bounded away from zero for 
$z\in[(\partial\Omega_1\cap\partial\Omega_6)\cup(\partial\Omega_2\cap\partial\Omega_5)]\cap(\mathbb{D}_\alpha\cup\mathbb{D}_\beta\cup\mathbb{D}_{b^{\star}})^c$.
\item[(c)] On the outer lenses outside of $\mathbb{D}_{b^{\star}}$:  The real 
part of $\mathcal{P}_3(z;\kappa)$ is negative and bounded away from zero 
for
$z\in[(\partial\Omega_1\cap\partial\Omega_2)\cup(\partial\Omega_5\cap\partial\Omega_6)]\cap\mathbb{D}_{b^{\star}}^c.$
\item[(d)] On the real axis outside of the outer lenses and 
$\mathbb{D}_{b^{\star}}$:  The real part of 
$\mathcal{P}_2(z;\kappa)$ is negative and bounded away from zero for 
$z\in(\partial\Omega_1\cap\partial\Omega_6)\cap\mathbb{D}_{b^{\star}}^c$.
\item[(e)] On the real axis inside $\mathbb{D}_{b^{\star}}$:  The real parts of 
$$\mathcal{P}_1(z;\kappa)-\kappa g_H(\zeta)+\kappa\log\zeta + \kappa\sum_{j=1}^k\frac{c_j^{(H)}}{\zeta^{2j}}  \quad \text{and} \quad \mathcal{P}_2(z;\kappa) + \kappa g_H(\zeta) - \kappa\log\zeta - \kappa\sum_{j=1}^k\frac{c_j^{(H)}}{\zeta^{2j}}$$ 
are negative and bounded away from zero for 
$z\in[(\partial\Omega_1\cap\partial\Omega_6)\cup(\partial\Omega_2\cap\partial\Omega_5)]\cap\mathbb{D}_{b^{\star}}$.
\end{itemize}
\end{lemma}
\begin{proof}
Parts (a) and (b) come from Proposition \ref{propdeform-sub} along with the fact that 
$\log(z-b^{\star})$ and $(z-b^{\star})^{-j}$, $j=1,\dots,2k$, 
is bounded outside of $\mathbb{D}_{b^{\star}}$.  Part (c) comes from 
combining Lemma \ref{sub-nolog-lemma}(c) with Proposition  \ref{propdeform-sub}  and the 
boundedness of $\log(z-b^{\star})$ and $(z-b^{\star})^{-j}$, 
$j=1,\dots,2k$.

  Part (d) follows from Lemma \ref{sub-nolog-lemma}(d), 
Proposition \ref{propdeform-sub} , and the boundedness of 
$\log(z-b^{\star})$ and $(z-b^{\star})^{-j}$, $j=1,\dots 2k$.

Finally, part (e) comes from (b) and (d) of Lemma 
\ref{sub-nolog-lemma} along with \eqref{sub-zeta} and Proposition \ref{propdeform-sub}.
\end{proof}
These results allow us to now bound the jumps ${\bf V^{(E)}}$ of the 
error problem.  Divide $\Gamma^{\bf(E)}$ into a compact component $\Gamma_C^{\bf(E)}$ and a 
noncompact component $\Gamma_N^{\bf(E)}$:
\eq
\begin{split}
\Gamma_C^{\bf(E)}:=&\partial\mathbb{D}_\alpha\cup\partial\mathbb{D}_\beta\cup\partial\mathbb{D}_{b^{\star}}\cup(\Gamma\cap\mathbb{D}_{b^{\star}}),\\
\Gamma_N^{\bf(E)}:=&\Gamma^{\bf(E)}\backslash\Gamma_C^{\bf(E)}.
\end{split}
\endeq
\begin{lemma}
\label{sub-VE-lemma}
In the subcritical regime, for large $n$:
\begin{itemize}
\item[(a)]  Outside the disks $\mathbb{D}_\alpha$, $\mathbb{D}_\beta$, 
and $\mathbb{D}_{b^{\star}}$:  There is a constant $c>0$ such that 
$${\bf V^{(E)}}(z;\kappa) = {\bf I} + \mathcal{O}(e^{-cn}), \quad z\in\Gamma_N^{\bf(E)}.$$
\item[(b)]  On the boundary of $\mathbb{D}_{b^{\star}}$:  
$${\bf V^{(E)}}(z;\kappa) = {\bf I} + \mathcal{O}\left(\ds n^{-(1-\gamma)/2}\right) + \mathcal{O}\left(\ds n^{-k-1+(k+2)\gamma} \right), \quad  z\in\partial\mathbb{D}_{b^{\star}}.$$
\item[(c)]  On the boundaries of $\mathbb{D}_\alpha$ and $\mathbb{D}_\beta$:  
$${\bf V^{(E)}}(z;\kappa) = {\bf I} + \mathcal{O}\left(\frac{1}{n}\right), \quad z\in\partial\mathbb{D}_\alpha\cup\partial\mathbb{D}_\beta.$$
\item[(d)]  Inside $\mathbb{D}_{b^{\star}}$:  There is a constant $c>0$ such that 
$${\bf V^{(E)}}(z;\kappa) = {\bf I} + \mathcal{O}\left(e^{-cn}\right), \quad z\in\Gamma\cap\mathbb{D}_{b^{\star}}.$$
\end{itemize}
\end{lemma}
\begin{proof}
Part (a) is the result of Lemma \ref{sub-Pmathcal-lemma}(a)--(d) and the 
boundedness of ${\bf \Psi}(z)$.  Part (b) is from  
\eqref{sub-zeta-order}, \eqref{sub-R-error}, and the 
boundedness of ${\bf \Psi}(z)$.  Part (c) is from the construction 
of the parametrices ${\bf P}_\text{\bf Ai}^{(\alpha)}(z)$ and 
${\bf P}_\text{\bf Ai}^{(\beta)}(z)$ (see, for instance, 
\cite{Deift:1999b}).  
\\
\indent
For part (d), consider the jumps \eqref{sub-VE-b*1}.  
By \eqref{sub-R} for ${\bf R}(\zeta)$,
\eq
\label{VE-bound-Gamma-b*1}
\begin{split}
{\bf V^{(E)}}(z;\kappa) = & {\bf \Psi}(z) e^{-\frac{r}{2}\ell_H{\bf\Lambda}_{23}} {\bf H}_{23-}(\zeta) e^{-r(g_H(\zeta)-\frac{\ell_H}{2}){\bf\Lambda}_{23}} \bpm 1 & (*)_{12} & (*)_{13} \\ 0 & 1 & 0 \\ 0 & 0 & 1 \epm \times \\
& \times e^{r(g_H(\zeta)-\frac{\ell_H}{2}){\bf\Lambda}_{23}}{\bf H}_{23-}(\zeta)^{-1}e^{\frac{r}{2}\ell_H{\bf\Lambda}_{23}}{\bf \Psi}(z)^{-1}, \quad z\in(\partial\Omega_1\cap\partial\Omega_6)\cap\mathbb{D}_{b^{\star}}, 
\end{split}
\endeq
and
\eq
\label{VE-bound-Gamma-b*2}
\begin{split}
{\bf V^{(E)}}(z;\kappa) = & {\bf \Psi}(z) e^{-\frac{r}{2}\ell_H{\bf\Lambda}_{23}} {\bf H}_{23-}(\zeta) e^{-r(g_H(\zeta)-\frac{\ell_H}{2}){\bf\Lambda}_{23}} \bpm 1 & (*)_{12} & 0 \\ 0 & 1 & 0 \\ 0 & 0 & 1 \epm \times \\
& \times e^{r(g_H(\zeta)-\frac{\ell_H}{2}){\bf\Lambda}_{23}}{\bf H}_{23-}(\zeta)^{-1}e^{\frac{r}{2}\ell_H{\bf\Lambda}_{23}}{\bf \Psi}(z)^{-1}, \quad z\in(\partial\Omega_2\cap\partial\Omega_5)\cap\mathbb{D}_{b^{\star}}, 
\end{split}
\endeq
wherein
\eq
\label{sub-VE-entries}
\begin{split}
(*)_{12} = & \exp\left(n\mathcal{P}_1(z;\kappa) - rg_H(\zeta) + r\log\zeta + r\sum_{j=1}^k\frac{c_j^{(H)}}{\zeta^{2j}}\right), \\
(*)_{13} = & \exp\left(n\mathcal{P}_2(z;\kappa)+rg_H(\zeta)-r\log\zeta - r\sum_{j=1}^k\frac{c_j^{(H)}}{\zeta^{2j}}\right).
\end{split}
\endeq
This along with Lemma \ref{sub-Pmathcal-lemma}(e) and the boundedness of 
${\bf \Psi}(z)$ in $\mathbb{D}_{b^{\star}}$ gives the result (d).
\end{proof}
We can now show that ${\bf E}(z)$ is asymptotically close to the identity.  
The proof of the following lemma follows that of Lemma \ref{super-E}:
\begin{lemma}
\label{sub-E}
In the subcritical regime, for $n$ large,
$${\bf E}(z) = {\bf I} + \mathcal{O}\left(n^{-(1-\gamma)/2}\right)$$
uniformly in $z$.
\end{lemma}


\subsection{The subcritical kernel and proof of Theorem \ref{theorem-sub-kernel}}

\begin{proof}[Proof of Theorem \ref{theorem-sub-kernel}]
Once again, recall the kernel (\ref{mop-kernel}):
\eq
K_n(x, y) = \frac{e^{-\frac{1}{2} n ( V(x) + V(y) )}}{2\pi i(x-y)} \left( \left[ {\bf Y}(y)^{-1} {\bf Y}(x) \right]_{21} + e^{n a y} \left[ {\bf Y}(y)^{-1}{\bf Y}(x)\right]_{31} \right).
\endeq
While the function ${\bf Y}(z)$ has a 
jump for $z\in\mathbb{D}_{b^{\star}}$, the first column of ${\bf Y}(z)$ does not 
(see the 
Riemann-Hilbert problem \eqref{rhp}); observing the Riemann Hilbert Problem  for ${\bf Y}^{-1}$ we also note that  the second and third {\em rows} of ${\bf Y}^{-1}$ are entire functions.  Therefore we can pick $x$ and $y$ to 
be in a convenient region.  We choose $x$ and $y$ to be in $\Omega_1$ as 
defined in Figure \ref{sublens}.  

From the transformation \eqref{super-Y-to-W} (which is the same in the subcritical case as noted at the beginning of Section \ref{sub-steepest}) and using the new definitions \eqref{sub-mathcalP1}, \eqref{sub-mathcalP2}, \eqref{sub-mathcalP3}, we see that 
\eq
\label{Yinv-Y21}
\left[{\bf Y}(y)^{-1}{\bf Y}(x)\right]_{21} = \left[{\bf W}(y)^{-1}{\bf W}(x)\right]_{21}\exp\left(n\left({\mathfrak g}(y)+{\mathfrak g}(x)
{-\delta V(y)} 
+\ell_1\right)\right),
\endeq
\eq
\label{Yinv-Y31}
\left[{\bf Y}(y)^{-1}{\bf Y}(x)\right]_{31} = \left[{\bf W}(y)^{-1}{\bf W}(x)\right]_{31}\exp\left(n\left({\mathfrak g}(x)+
{ \delta V(y)} 
+
{l_2-\eta}\right)\right)
\endeq
for $x$ and $y$ in $\Omega_1$.  As in the supercritical case, we have 
\eq
{\bf W}(z) ={\left ({\bf I} + \mathcal{O}\left(n^{-(1-\gamma)/2}\right)\right)} {\bf \Psi}(z){\bf R}(\zeta(z))
\endeq
and
\eq
\label{sub-PsiPsiinf}
{\bf \Psi}(y)^{-1}{\bf \Psi}(x) = {\bf I} + \mathcal{O}\left((\zeta_x-\zeta_y)\kappa^{1/2}\right).
\endeq
Define $\mathcal{Q}_i(z;\kappa)$ to be 
$\mathcal{P}_i(z;\kappa)$ without the logarithm or pole terms:
\begin{eqnarray}
{\mathcal Q}_1(z;\kappa) & := & -V(z) + 2{\mathfrak g}(z;\kappa) + \ell_1, \\
{\mathcal Q}_2(z;\kappa) & := & -V(z) + az + {\mathfrak g}(z;\kappa) + 
{l_2}, \\
{\mathcal Q}_3(z;\kappa) & := & az - {\mathfrak g}(z;\kappa) + 
{l_2-\ell_1}.
\end{eqnarray}
Recall \eqref{sub-H23} that 
\be
\label{342}
{\bf R}(\zeta)  = \mathcal O(1) {\rm e}^{-n \le( \kappa \sum_{j=1}^K \frac {c_j^{(H)}}{\zeta^{j}} + \kappa \ln \zeta - \frac \kappa 2 \ell_H\ri)\mathbf \Lambda_{23}} \ \text{as} \  \zeta\to 0.
\ee
Now combining \eqref{Yinv-Y21}, \eqref{sub-PsiPsiinf}, \eqref{342}
gives
\eq
e^{-\frac{n}{2}(V(x)+V(y))}\left[{\bf Y}(y)^{-1}{\bf Y}(x)\right]_{21} = \mathcal{O}((x-y)n^{-(1-\gamma)/2})e^{n(*)}
\endeq
and
\eq
e^{-\frac{n}{2}(V(x)+V(y))+nay}\left[{\bf Y}(y)^{-1}{\bf Y}(x)\right]_{31} = \mathcal{O}((x-y)n^{-(1-\gamma)/2}) e^{n(**)}.
\endeq
Rearranging the terms from \eqref{38} and \eqref{312} we find 
\be
\delta V(z) - \kappa \ln \zeta - \kappa\sum_{j=1}^K \frac {c^{(H)}_j} {\zeta^j}  = \frac 1 2 \le(
\mathfrak g(z) - a z - \frac \kappa 2 (\zeta - \zeta_0)^2 + \kappa \ln \kappa  + \mathfrak b -l_2  + \ell_1 
\ri)\label{346}.
\ee
Using \eqref{346} we can rewrite
%
%
\bea
(\star) = 
\frac  1 2 \mathcal Q_1(x) + \frac 1 2 \mathcal Q_2(y) + \frac \kappa 4 (\zeta_y-\zeta_0)^2 -  \kappa \ln \sqrt \kappa -\frac {\mathfrak b }{2} = 
\frac 1 2 P_1(x) + \frac 1 2 P_2(y)  + \mathcal O(\kappa \ln \kappa)
\eea
and
\bea
(\star \star) = \frac 1 2 \mathcal Q_1(x) + \frac 12 \mathcal Q_2(y)  - \frac \kappa 4 (\zeta_y - \zeta_0)^2 - \frac \eta 2 + \kappa \ell_H=  \frac 12 P_1(x) + \frac 1 2 P_2(y) + \mathcal O(\kappa \ln \kappa).
\eea
%
Here we have used  Proposition \ref{propdeform-sub} to convert $\mathfrak{g}$ to $g$. 
Since $\Re P_1 (b^\star)< 0$ and also $\Re P_2(b^\star)<0$, the theorem follows.
\end{proof}

\appendix
\renewcommand{\theequation}{\Alph{section}-\arabic{equation}}

\section{The detailed analysis of Theorem \ref{rhoprop} }
\label{confmap-appendix}

In this appendix we prove the existence of the local change of variables 
used in the supercritical and subcritical cases.  We also demonstrate how 
the change of variables can be computed explicitly termwise.

\subsection{Background material}
We start recalling that the space $\mathcal H(\mathbb D(r))$ of holomorphic functions on an open connected domain $\mathbb D$ (a disk of radius $r$ for simplicity) is a Banach space with respect to  the sup norm.

The theorem of existence and uniqueness for ODEs can be extended  to any  Banach space $\mathcal M$. A sufficient condition for the integrability  is the Lipshitz property, namely   that we are given a (time-dependent) vector field $\mathcal V: \mathcal M\times J \to T\mathcal M$  which is {\bf jointly continuous} and {\bf Lipshitz}. 
Let 
\bea
&& \Omega_1:= \{\zeta:\mathbb D(r)\to \C,\  \zeta(0)=0, \|\zeta\|_\infty<\infty, \zeta \hbox{ {\bf univalent}}\}
\eea
and
\bea
&& \Omega:= \{\zeta:\mathbb D(r)\to \C,\  \zeta(0)=0, \|\zeta\|_\infty<\infty\}.
\eea

\begin{lemma}
\label{lemmaLipshitz}
The evaluation map of the inverse $\rho^{-1}$ at a point is locally Lipshitz on $\Omega_1$. More precisely:  
\be
\forall \zeta_0\in \Omega_1\ \exists C, \rho, S>0  \ s.t. \ \forall \xi\in \C,\  |\xi|<\rho\ \  \forall \zeta_1,\zeta_2 \in B_S(\zeta_0)\subset \Omega_1
\ee 
\be
|\zeta_1^{-1}(\xi) - \zeta_2^{-1}(\xi)| \leq C \|\zeta_1 -\zeta_2\|_\infty = C \ \sup_{z\in \mathbb D(r)} |\zeta_1(z)-\zeta_2(z)|\ .
\ee
\end{lemma}
\begin{proof}
Note that $\Omega_1$ is an {\bf open} subset of the Banach vector space  of bounded analytic functions on $\mathbb D(r)$ containing the identity map.
Therefore the Banach ball of radius $S>0$ centered at $\zeta_0\in \Omega_1$ lies all within $\Omega_1$ for sufficiently small $S$.

 First we note that the forward map is locally Lipshitz; that is, let $z_0\in \mathbb D(r/2)$, then  any of the functionals $\zeta^{(n)} (z_0)$ are Lipshitz 
\bea
|\zeta_1^{(n)} (z_0) - \zeta_2^{(n)}(z_0)| = \le|
\frac {n!}{2i\pi} \oint_{|z|=2/3r} \frac {(\zeta_1(z)-\zeta_2(z))\d z}  {(z-z_0)^{n+1}} \ri| \leq \frac {1}{2\pi \le(\frac 2 3 - \frac 1 2\ri)^{n+1} r^n }  \|\zeta_1-\zeta_2\|\ .
\eea

Let now $\zeta_0(z) \in \Omega_1$ be univalent and let $0< 3\rho:= \inf_{|z|=r/2} |\zeta_0(z)|$.  Let $\xi$ be such that $|\xi|<\rho$.

Let $\zeta_1,  \zeta_2$ be two  maps in a ball around $\zeta_0$ of radius $\rho$ ($\|\zeta_j-\zeta_0\|<\rho$) and consider (all integrations are on $|z|=\frac 12 r$)
\bea
|z_0-\widetilde z_0|:= \le| \zeta_1^{-1}(\xi) - \zeta_2^{-1}(\xi)\ri| = \le|\frac 1 {2i\pi} \oint \frac {z \zeta'_1(z)\d z}{\zeta_1(z)-\xi} -  \frac 1 {2i\pi} \oint  \frac {z  \zeta_2'(z)\d z}{ \zeta_2(z)-\xi} \ri|  \\
\leq
\frac 1{2\pi} \oint \le|\frac {z\le[ (\zeta_1'(z)( \zeta_2 - \xi)-\zeta_2'(z) (\zeta_1 -\xi)\ri]\d z}{(\zeta_1-\xi)( \zeta_2-\xi)}  \ri|
\\
=
\frac 1{2\pi} \oint \le|\frac {z\le[ \xi (\zeta_2'-\zeta_1') + \zeta'_1 (\zeta_2-\zeta_1) + (\zeta'_1 -\zeta_2') \zeta_1 \ri]\d z}{(\zeta_1-\xi)( \zeta_2-\xi)}  \ri| \ .
\eea
Since the derivative evaluation on the circle $z=r/2$ is uniformly Lipshitz, the above can be easily estimated by 
\be
\frac 1{2\pi} \oint \le|\frac {z\le[ \xi (\zeta_2'-\zeta_1') + \zeta'_1 (\zeta_2-\zeta_1) + (\zeta'_1 -\zeta_2') \zeta_1 \ri]\d z}{(\zeta_1-\xi)( \zeta_2-\xi)}  \ri| \leq  \frac {C r}{\ds \inf_{|z|=r/2} |\zeta_1-\xi| \inf_{|z|=r/2}  |\zeta_2-\xi|}\|\zeta_1-\zeta_2\|
\ee
where $C = \sup \{|\xi|, \sup_{|z|=r/2} |\zeta_1'| , \sup_{|z|=r/2} |\zeta_1|\}$. 
Since $\zeta_j$ are less than $\rho$ away from $\zeta_0$ we have that the two infima in the denominator are at least $\rho$, since $|\xi|<\rho$ and 
\be
\inf_{|z|=r/2} |\zeta_j|\geq  \inf_{|z|=r/2} |\zeta_0|-\rho = 3\rho-\rho=2\rho.
\ee
\end{proof}

\subsection  {Complete proof of Theorem \texorpdfstring{\ref{rhoprop}  for the case $K=0$}{rhoprop}}
 \begin{proof}
 In this case there are no $\vec \delta$'s and no $\vec \gamma$'s; it should become clear that the general proof presents only notational complications, but is amenable to the same logic and hence the details are omitted. We also omit  explicit reference to the dependence of $f(z)$ on $\kappa$ for brevity.
We want to have
\be
  - f(z) +2 \kappa  \ln z  =-\frac 1 2 (\rho(z)-{\mathfrak a} )^2 + \mathfrak b  + 2\kappa  \ln \rho(z) \ ,
\label{identity0}
\ee
where the goal is now to show that $\mathfrak b=\mathfrak b(\kappa)$, ${\mathfrak a} ={\mathfrak a} (\kappa)$ and $\rho = \rho(z;\kappa)$ are all analytic functions of $\kappa$, with $\rho$ being univalent in a neighborhood of $z=0$ and mapping the origin to the origin.
Consider the differentiation of the above identity with respect to  $\kappa$:
\be
- \dot f  (z)+2 \ln z =  \le( {\mathfrak a}  - \rho + \frac \kappa  \rho \ri) \dot \rho + (\rho-{\mathfrak a} )\dot {\mathfrak a}   + \dot \beta+2 \ln \rho \ .
\ee
Solve for $\dot \rho$ and we find 
\be
\dot \rho(z;\kappa) = \rho \frac {(\rho-{\mathfrak a} )\dot {\mathfrak a}  + \dot {\mathfrak b} +2 \ln\le(\frac \rho z\ri)  + \dot f(z)} { \rho^2  -  {\mathfrak a} \rho - \kappa }\ . \label{1.5} 
\ee 

We want to view this equation as defining a vector field on a suitable Banach space that we define presently.
Let $\Omega_{1}$ be the Banach {\bf manifold} of {\bf univalent, analytic functions} $\rho:\mathbb D(r)\to \C$ which fix the origin $\rho(0)=0$; this is a closed Banach submanifold of all univalent analytic functions because the evaluation map is continuous.  One only has to verify that if $\rho_n$ is a sequence of univalent analytic functions on $\mathbb D(r)$ converging in the sup-norm, the limit exists and it is still univalent.  Define now 
\be
\mathcal M:= \Omega_1 \times \C^2 = \le \{{\bf p} =  (\rho,{\mathfrak a} ,{\mathfrak b}),\ \ \zeta\in \Omega_1, \ \ {\mathfrak a} ,{\mathfrak b}\in \C \ri \}\ .
\ee

Formula (\ref{1.5}) defines a vector field on $\mathcal M$: we  will first explain it in coarse terms and then refine the details.

The denominator to (\ref{1.5}) has two roots $\rho_1({\mathfrak a} ,\kappa),\rho_2({\mathfrak a}  ,\kappa)$; since $\dot \rho$ must be also analytic, we must impose that the numerator vanishes at the same points, and hence 
\bea
\dot {\mathfrak a}  &&\hspace{-18pt}=A:= \frac{ \det\begin{bmatrix} -\dot f(z_1) -  2\ln (\rho_1/z_1) & 1 \\ -\dot f(z_2) - 2\ln (\rho_2/z_2) & 1 \end{bmatrix}}{ \det\begin{bmatrix} \rho_1 - {\mathfrak a}  & 1\\
\rho_2-{\mathfrak a}  &1 \end{bmatrix}} 
\label{Adef}
\\
\dot {\mathfrak b} &&\hspace{-18pt}=B:= \frac{ \det\begin{bmatrix} (\rho_1-{\mathfrak a} )  & - \dot f(z_1) -  2\ln (\rho_1/z_1)  \\ (\rho_2-{\mathfrak a} ) &  - \dot f(z_2) -2\ln (\rho_2/z_2) \end{bmatrix}}{ \det\begin{bmatrix} \rho_1-\mathfrak a & 1\\
\rho_2 -\mathfrak a&1 \end{bmatrix}}
\label{Bdef}\\
&& \rho_{1,2} := \frac {{\mathfrak a}  \pm \sqrt{{\mathfrak a} ^2+4\kappa}}{2}.
\eea

 Here $z_1,z_2$ are the counterimages of $\rho_1,\rho_2$,  $z_j :=\rho^{-1}(\rho_j)$. Note that the expressions have analytic continuations to the case $\rho_1 = \rho_2$: indeed they are symmetric functions of the roots and therefore they can be expressed in terms  of analytic functions of ${\mathfrak a} ,\kappa$ (which play the role of elementary symmetric polynomials in the roots).
Therefore we consider the (time dependent) vector field $\mathcal V$ on the manifold $\mathcal M$ 
\bea
\label{defetavect}
\mathcal V([(\rho,{\mathfrak a} ,{\mathfrak b},\kappa)]) &&\hspace{-18pt} = \le(\eta(z), A, B \ri):= \le( \rho \frac {(\rho-{\mathfrak a} )A + B+ 2\ln\le(\frac \rho z\ri) + \dot f(z)} { \rho^2  -  {\mathfrak a} \rho - \kappa }, A, B
 \ri).
\eea
It is to be pointed out that $\eta(z) = \eta([\rho,{\mathfrak a} ,{\mathfrak b},\kappa]; z)$ is a tangent vector to $\Omega_1$, namely $\eta([\rho,{\mathfrak a} ,{\mathfrak b},\kappa ]; 0)\equiv 0$.

The initial condition for the vector field is 
\be
{\bf p}_0 = \le (\rho(z;0),{\mathfrak a} (0),{\mathfrak b}(0)\ri)  =\le(  \sqrt{2 f(z)}, 0 ,0 \ri).
\ee

Therefore the proof shall follow if we show that the vector field $\mathcal V$ is integrable in some neighborhood of the initial point and for sufficiently small values of $\kappa$ and for this to hold it is sufficient to verify the Lipshitz property.

\paragraph{Lipshitz property for $\mathcal V$.}
 To complete the proof it is sufficient to show that the vector field is locally Lipshitz in a Banach neighborhood of the initial condition. 
The initial $\rho (z)$ is univalent in a small disk --say $\mathbb D( 2 r_0)$-- around $z=0$ because $\rho '(0)\neq 0$.

By simple continuity arguments in the sup norm, there is a sup-neighborhood $\mathcal U$ of $\rho$ consisting of univalent functions on $\mathbb D(r_0)$.

We therefore shall restrict ${\mathfrak a} ,  \kappa $ in such a way that $|\rho_j|< r_0$; this guarantees that we can {\em define} the components of $\mathcal V$. It is also quite clear that the restriction $|\rho_j({\mathfrak a} ,\kappa)|<r_0$ contains a polydisk in ${\mathfrak a}  , \kappa $ (here ${\mathfrak b}$ is unrestricted). 
For example we can require
\be
|{\mathfrak a}  |< \frac {r_0}{10}\ ,\ |\kappa|< \frac {{r_0}^2} {100}\ .
\ee

Thus, the neighborhood of the initial condition $\{{\bf p}_0\}\times\{0\}$ that  we will analyze is 
\be
\mathfrak V:= \le\{({\bf p},\kappa) =(\rho, {\mathfrak a}  ,{\mathfrak b},\kappa  ) \in \mathcal M\times \R:\ \ \rho \in \mathcal U\ ,\ \ |{\mathfrak a}  |< \frac {r_0}{10}\ ,\ |\kappa|< \frac {{r_0}^2} {100}  \ri\}
\ee
The goal is now to prove that $\mathcal V$ is Lipshitz on $\mathfrak V$ for any fixed $t$. The fact that $A , B$ in (\ref{Adef}, \ref{Bdef}) are Lipshitz functions follows from Lemma \ref{lemmaLipshitz}. As for the first component, we have recall that the product of two Lipshitz {\bf bounded} functions is Lipshitz, as well as the ratio if the denominator is bounded away from zero. By construction $\eta(z)$ (\ref{defetavect})  is analytic and hence its sup-norm is achieved on the boundary of $\mathbb D(r)$ (by the maximum modulus theorem).  By the restrictions we made on $|{\mathfrak a}  |$ and $|\kappa|$, the denominator is bounded away from zero on $\pa\mathbb D(r)$ and uniformly so with respect to the choice of $\rho$ in the  Banach neighborhood $\mathcal U$ of ${\mathfrak a} $. On the other hand, the numerator is clearly Lipshitz.
\end{proof}





%

\bibliographystyle{alpha}

\begin{thebibliography}{99}

\bibitem{Adler:2009a}
M. Adler, J. Del\'epine, and P. van Moerbeke,
Dyson's nonintersecting Brownian motions with a few outliers.
\textit{Commun. Pure Appl. Math.}
\textbf{62}
(2009),
334--395.

\bibitem{Adler:2009b}
M. Adler, N. Orantin, and P. van Moerbeke,
Universality of the Pearcey process.
arXiv:0901.4520 [math-ph] (2009).

\bibitem{Aptekarev:2005}
A. Aptekarev, P. Bleher, and A. Kuijlaars,
Large $n$ limit of Gaussian random matrices with external source, part II.
\textit{Comm. Math. Phys.}
\textbf{259}
(2005),
367--389.

\bibitem{Baik:2005}
J. Baik, G. Ben Arous, and S. P\'ech\'e,
Phase transition of the largest eigenvalue for non-null complex sample covariance matrices.
\textit{Ann. Probab.}
\textbf{33}
(2005),
1643--1697.

\bibitem{private}
J. Baik, and D. Wang.  (To be published; Private communication.)


\bibitem{Baik:2006}
J. Baik,
Painleve formulas of the limiting distributions for non-null complex sample covariance matrices.
\textit{Duke Math. J.}
\textbf{133}
(2006),
205--235.

\bibitem{Baik:2008}
J. Baik,
On the Christoffel-Darboux kernel for random Hermitian matrices with external source.
\textit{Comput. Methods Funct. Theory} \textbf{9} (2009), 455--471

\bibitem{Bertola:2010}
M. Bertola, R. Buckingham, S.Y. Lee, and V. Pierce,
Spectra of random Hermitian matrices with a small-rank external source:  The critical regimes.
Preprint 2010.


\bibitem{Bertola:2009a}
M. Bertola and S. Lee,
First colonization of a spectral outpost in random matrix theory.
\textit{Constr. Approx.} \textbf{30} (2009), no. 2, 225--263.

\bibitem{Bertola:2008}
M. Bertola and S. Lee,
First colonization of a hard-edge in random matrix theory.
\textit{Constr. Approx.} 
\textbf{31} (2010), no. 2, 231--257.



\bibitem{Bertola:2009}
M. Bertola, S. Lee, and M. Yo,
Mesoscopic colonization of a spectral band.
\textit{J. Phys. A}
\textbf{ 42} (2009) no. 41, 415204, 17.

\bibitem{Bleher:2010}
P. Bleher, S. Delvaux, and A. Kuijlaars,
Random matrix model with external source and a constrained vector equilibrium problem.
arXiv:1001.1238 [math-ph] (2010).

\bibitem{Bleher:2004a}
P. Bleher and A. Kuijlaars,
Random matrices with external source and multiple orthogonal polynomials.
\textit{Int. Mat. Res. Not. IMRN}
\textbf{3}
(2004),
109--129.

\bibitem{Bleher:2004b}
P. Bleher and A. Kuijlaars,
Large $n$ limit of Gaussian random matrices with external source, part I.
\textit{Comm. Math. Phys.}
\textbf{252}
(2004),
43--76.

\bibitem{Bleher:2007}
P. Bleher and A. Kuijlaars,
Large $n$ limit of Gaussian random matrices with external source, part III:  double scaling limit.
\textit{Comm. Math. Phys.}
\textbf{270}
(2007),
481--517.

\bibitem{Brezin:1996}
E. Br\'ezin and S. Hikami,
Correlations of nearby levels induced by a random potential.
\textit{Nuclear Phys. B}
\textbf{479}
(1996),
697--706.

\bibitem{Brezin:1997a}
E. Br\'ezin and S. Hikami,
Spectral form factor in a random matrix theory.
\textit{Phys. Rev. E}
\textbf{55}
(1997),
4067--4083.

\bibitem{Brezin:1997b}
E. Br\'ezin and S. Hikami,
Extension of level-spacing universality.
\textit{Phys. Rev. E}
\textbf{56}
(1997),
264--269.

\bibitem{Brezin:1998}
E. Br\'ezin and S. Hikami,
Level spacing of random matrices in an external source.
\textit{Phys. Rev. E}
\textbf{58}
(1998),
7176--7185.

\bibitem{Daems:2007}
E. Daems and A. Kuijlaars, 
Multiple orthogonal polynomials of mixed type and non-intersecting Brownian motions.
\textit{J. Approx. Theory} 
\textbf{146} 
(2007), 
92--114.


\bibitem{Deift:1998-book}
P. Deift, 
\textit{Orthogonal Polynomials and Random Matrices:  a Riemann-Hilbert Approach}.
American Mathematial Society,
Providence,
1998.

\bibitem{Deift:1998a}
P. Deift, T. Kriecherbauer, K. McLaughlin,
New results on the equilibrium measure for logarithmic potentials in the presence of an external field.
\textit{J. Approx. Theory}
{\bf 95} (1998), 
 388--475. 

\bibitem{Deift:1999a}
P. Deift, T. Kriecherbauer, K. McLaughlin, S. Venakides, and X. Zhou,
Uniform asymptotics for polynomials orthogonal with respect to varying exponential weights and applications to universality questions in random matrix theory.
\textit{Comm. Pure Appl. Math.}
\textbf{52}
(1999),
1335--1425.

\bibitem{Deift:1999b}
P. Deift, T. Kriecherbauer, K. McLaughlin, S. Venakides, and X. Zhou,
Strong asymptotics of orthogonal polynomials with respect to exponential weights.
\textit{Comm. Pure Appl. Math.}
\textbf{52}
(1999),
1491--1552.

\bibitem{Delvaux:2010}
S. Delvaux, and A. Kuijlaars,
A graph-based equilibrium problem for the limiting distribution of nonintersecting Brownian motions at low temperature.
\textit{Constr. Approx.} (to be published)
arXiv:0907.2310 [math.CV] (2009).


\bibitem{Ercolani:2003}
N. Ercolani and K. McLaughlin,
Asymptotics of the partition function for random matrices via Riemann-Hilbert techniques and applications to graphical enumeration.  
\textit{Int. Math. Res. Not.}  
\textbf{14}
(2003),  
755--820.

\bibitem{Fokas:1992}
A. Fokas, A. Its, and A. Kitaev,
The isomonodromy approach to matrix models in 2D quantum gravity.
\textit{Comm. Math. Phys.}
\textbf{147}
(1992),
395--430.

\bibitem{McLaughlin:2007}
K. McLaughlin,
Asymptotic analysis of random matrices with external source and a family of algebraic curves.
\textit{Nonlinearity}
\textbf{20}
(2007),
1547--1571.

\bibitem{Peche:2006}
S. P\'ech\'e,
The largest eigenvalue of small rank perturbations of Hermitian random matrices.
\textit{Probab. Theory Related Fields}
\textbf{134}
(2006),
127--173.

\bibitem{Tracy:1994}
C. Tracy and H. Widom,
Level-spacing distributions and the Airy kernel.
\textit{Comm. Math. Phys.}
\textbf{159}
(1994),
151--174.

\bibitem{Wigner:1951}
E. Wigner,
On the statistical distribution of the widths and spacings of nuclear resonance levels.
\textit{Proc. Cambridge Philos. Soc.}
\textbf{47}
(1951),
790--798.

\bibitem{Zinn-Justin:1997}
P. Zinn-Justin,
Random Hermitian matrices in an external field.
\textit{Nuclear Phys. B}
\textbf{497}
(1997),
725--732.

\bibitem{Zinn-Justin:1998}
P. Zinn-Justin,
Universality of correlation functions of Hermitian random matrices in an external field.
\textit{Comm. Math. Phys.}
\textbf{194}
(1998),
631--650.

\end{thebibliography}

\end{document}